\documentclass[journal,onecolumn,draftclsnofoot,twoside]{IEEEtran} %
\usepackage[top=1in,bottom=1in,left=1in,right=1in]{geometry}

\usepackage{amsmath,amsfonts,amssymb} 
\usepackage{amsthm} %

\usepackage{color}
\usepackage{mathtools}
\usepackage{hhline}
\usepackage{verbatim}
\usepackage[dvipsnames]{xcolor}
\usepackage{enumerate} %
\usepackage{changepage}   %
\usepackage{subcaption}

\usepackage{bbm} %

\usepackage{cite}

\theoremstyle{plain}

\newtheorem{theorem}{Theorem}
\newtheorem{lemma}{Lemma}
\newtheorem{corollary}{Corollary}

\newtheorem{remark}{Remark}

\renewcommand{\baselinestretch}{1} %

\newcommand{\EE}{\mathbb{E}}

\newcommand{\PP}{\mathbb{P}}

\newcommand{\RR}{\mathbb{R}}

\newcommand{\calA}{\mathcal{A}}
\newcommand{\calB}{\mathcal{B}}
\newcommand{\calC}{\mathcal{C}}

\newcommand{\calE}{\mathcal{E}}

\newcommand{\calG}{\mathcal{G}}

\newcommand{\calN}{\mathcal{N}}

\newcommand{\calS}{\mathcal{S}}

\newcommand{\calV}{\mathcal{V}}

\newcommand{\QED}{\hfill $\square$}

\newcommand{\Sum}[2]{\sum\limits_{#1}^{#2}}

\DeclarePairedDelimiter\ceil{\lceil}{\rceil}

\newcommand{\lp}{\left(}
\newcommand{\rp}{\right)}
\newcommand{\lb}{\left[}
\newcommand{\rb}{\right]}
\newcommand{\lbp}{\left\{}
\newcommand{\rbp}{\right\}}

\newcommand{\bbm}{\mathbbm}

\newcommand{\msf}{\mathsf}

\newcommand{\eqDef}{\triangleq}
\newcommand{\diid}{\overset{\text{i.i.d.}}{\sim}}

\usepackage{hyperref}
\title{Adaptive Group Testing on Networks with Community Structure: \\The Stochastic Block Model}
\author{Surin Ahn, Wei-Ning Chen, and Ayfer {\"O}zg{\"u}r

	\thanks{Surin~Ahn, Wei-Ning Chen, and Ayfer~{\"O}zg{\"u}r are with the Department of Electrical Engineering, Stanford University, Stanford, CA 94305 USA (e-mail: surinahn@stanford.edu; wnchen@stanford.edu; aozgur@stanford.edu).}
	
	\thanks{This work was presented in part at the \emph{2021 IEEE International Symposium on Information Theory (ISIT)}. }
	}

\begin{document}

	\maketitle 
	
	\begin{abstract}
   Group testing was conceived during World War II to identify soldiers infected with syphilis using as few tests as possible, and it has attracted renewed interest during the COVID-19 pandemic. A long-standing assumption in the probabilistic variant of the group testing problem is that individuals are infected by the disease \textit{independently}.
   However, this assumption rarely holds in practice, as diseases often spread through interactions between individuals and therefore cause infections to be correlated. Inspired by characteristics of COVID-19 and other infectious diseases, we introduce an infection model over networks which generalizes the traditional i.i.d. model from probabilistic group testing. Under this model, we ask whether knowledge of the network structure can be leveraged to perform group testing more efficiently, focusing specifically on community-structured graphs drawn from the stochastic block model. We prove that a simple community-aware algorithm outperforms the baseline binary splitting algorithm when the model parameters are conducive to ``strong community structure.'' Moreover, our novel lower bounds imply that the community-aware algorithm is order-optimal in certain parameter regimes. We extend our bounds to the noisy setting and support our results with numerical experiments. 
	\end{abstract}
	
	
	\section{Introduction} 
	Identifying individuals who are infected by a disease is crucial for curbing epidemics and ensuring the well-being of society. However, due to high costs or limited resources, it is often infeasible to test every member of the population individually. During World War II, when the U.S. military sought to identify soldiers infected with syphilis, Dorfman  introduced the breakthrough concept of \textit{group testing} \cite{dorfman1943detection}. He showed that by testing \textit{groups} or \textit{pools} of samples, the infected people in a population of size $n$ can be identified with far fewer than $n$ tests. The key insight was that if the infected population is sparse, then each pooled test is likely to produce a negative result, in which case all individuals included in the test can simultaneously be deemed healthy. Today, group testing strategies are actively being used in the COVID-19 pandemic to identify infected individuals in an efficient and cost-effective manner \cite{ellenberg2020covid, mallapaty2020covid, cdc2020covid, hogan2020sample}. 
	There has also been a recent influx of papers which seek to improve or better understand group testing for COVID-19, e.g., \cite{mentus2020analysis, verdun2021group, theagarajan2020group, eberhardt2020multi, ghosh2020compressed, cohen2020multi, abraham2020bloom, nikolopoulos2020community, gabrys2020ac, zhu2020noisy, goenka2020contact, aldridge2020conservative}. 

    Dorfman’s seminal work and many subsequent works by other authors \cite{hwang1975generalized, wolf1985born, sobel1959group, berger2002asymptotic, luo2008neighbor, li2014group, kealy2014capacity}
    assume that the disease infects individuals in a statistically independent fashion. The simplest and most widely studied case, known as the \emph{i.i.d. model} or \emph{binomial model}, assumes individuals are infected independently with some common probability $p$.\footnote{ In a related, commonly studied probabilistic model---often called the \emph{combinatorial prior model} or the \emph{hypergeometric model }---it is assumed that a random set of $d$ individuals out of $n$ are infected according to some distribution (typically uniform) over all ${n \choose d}$ possibilities \cite{hwang1981hypergeometric, aldridge2014group, chan2014non, mazumdar2016nonadaptive, scarlett2016converse, cai2017efficient, scarlett2018noisy, lee2019saffron, inan2019optimality, inan2020strongly}. While this is slightly different from an i.i.d. assumption, it is still somewhat simplistic and fails to capture any dependencies that may exist between individuals.} However, this assumption of independence rarely holds in practice. Diseases typically spread through \textit{interactions} between individuals (e.g., familial, work-related, or other social interactions), thereby inducing correlated infections. It is thus natural to ask whether exploiting information about this connectivity structure can lead to more efficient group testing strategies. This problem is especially timely given the critical role that group testing has played in the COVID-19 pandemic, and given that the disease is known to spread between individuals in close contact with each other.%
	
    In this paper, we contribute to the nascent area of ``group testing under correlations'' by investigating whether knowledge of the \textit{interaction network} dictating the spread of the disease can be leveraged to perform pooled testing more efficiently. We introduce a novel community-oriented infection model, called the \textit{stochastic block infection model} ($\msf{SBIM}$), which generalizes the standard i.i.d. model to a setting in which the disease can be transmitted between individuals. Our model is equivalent to a certain graph-based infection spread mechanism operating upon the well-known \textit{stochastic block model} (SBM) for random graphs. For decades, the SBM has been utilized across the social, biological, and information sciences as a very simple yet natural way to model community structure in probabilistic networks.

	On the algorithmic side, we consider \textit{adaptive} group testing schemes, where the design of each test can be informed by the previous test results. We compare two different schemes: the standard \textit{binary splitting} algorithm \cite{du2000combinatorial} which is oblivious to the underlying network structure, and a simple \textit{community-aware} algorithm which essentially performs two stages of binary splitting: the first stage identifies the communities containing at least one infected member, and the second stage performs more fine-grained testing within the infected communities.
	We give precise upper bounds on the expected number of tests performed by each algorithm. Crucially, we show that when the model parameters yield ``strong community structure'' (in which case the disease is much more likely to be transmitted within a community than between communities), the community-aware algorithm's average complexity is asymptotically strictly better than that of binary splitting.
	Furthermore, we derive novel information-theoretic lower bounds that apply to all adaptive strategies and imply the order-optimality of the community-aware algorithm in certain parameter regimes.%
	We then extend our algorithms and bounds to the noisy setting---in which the test outcomes are passed through a binary symmetric channel---and find that the presence of noise does not affect the relative gains of using a community-oriented approach. Finally, we corroborate our results with numerical experiments. To the best of our knowledge, this is the first thorough characterization of the complexity of adaptive group testing in a networked setting. 
	
We note that the underlying principles of this paper may be relevant to numerous settings beyond epidemiology. In the past, group testing has been successfully applied to diverse domains including wireless communications \cite{wolf1985born, berger1984random, luo2008neighbor, inan2017sparse, inan2018energy, inan2019group, inan2019sparse, cohen2020efficient}, machine learning \cite{ubaru2017multilabel, zhou2014parallel, malioutov2013exact}, signal processing \cite{gilbert2008group, cohen2019serial}, and the analysis of data streams \cite{cormode2005s, emad2014poisson}. In these settings and others, there may be a natural ``clustering'' of the population into different subgroups which can inform the design of better group testing strategies, i.e., be exploited as ``side information.'' For example, devices which are closer together in a multiple access network may tend to be active or inactive at the same time. Exploring the potential applications of network-oriented group testing to these types of problems is of great interest.

	\paragraph{Related Works} 
	In \textit{graph-constrained} group testing \cite{harvey2007non, cheraghchi2012graph, karbasi2012sequential, spang2018unconstraining}, the tests must conform to a given network topology. For example, if the objective is to identify faulty links in a communication network by sending diagnostic packets, then each test must correspond to a valid path in the network. By contrast, our problem setup permits arbitrary tests, but we ask whether \textit{knowledge} of the interaction network can help reduce the number of tests.
	
	There is a rich literature on adaptive group testing dating back to the early work of Dorfman \cite{dorfman1943detection} and others \cite{sterrett1957detection, li1962sequential, finucan1964blood, hwang1972method, sobel1959group, sobel1966binomial, hu1981boundary, riccio2000sharper}, with several important results having emerged in recent years, e.g., \cite{allemann2013efficient}, \cite{baldassini2013capacity}, \cite{scarlett2018noisy}, \cite{scarlett2019efficient}, \cite{aldridge2019rates}, \cite{coja2021optimal}. These works focus on relatively simple combinatorial or probabilistic models.  A few prior works have departed from these standard models by assuming that infections occur independently with non-identical prior probabilities \cite{hwang1975generalized, li2014group, kealy2014capacity}. However, our paper pertains to the fully non-i.i.d. case in which infections can be correlated with potentially different priors, depending on the network structure.\footnote{For the sake of obtaining comprehensive results, we focus on a symmetric model in which infections are correlated and identically distributed. However, the general infection model that we propose in Section~\ref{subsec:SBM_equivalence} is fully non-i.i.d. } 
	
	The idea of community-aware group testing was first explored in \cite{nikolopoulos2020community}, which assumed the population is partitioned into disjoint ``families'' and that the disease spreads in two stages with independent infections at each stage. 
	Our work considers an infection mechanism which similarly operates in two stages but is designed to model the interaction-based transmissions by which diseases often spread in reality. %
	Finally, we would like to acknowledge a number of independent and concurrent works related to community-aware group testing \cite{nikolopoulos2020group, bertolotti2020network, arasli2020group, goenka2020contact, lin2020positively}.
	
	\paragraph{Notation}
	Let $[n] \triangleq \{1,2,\ldots, n\}$. We denote by $n, k$, and $m \triangleq \frac{n}{k}$ the size of the population, size of each community, and number of communities, respectively. $X \triangleq (X_1,\ldots, X_n) \in \{0,1\}^n $ is the infection status vector, where $X_i = 1$ iff the $i^\text{th}$ individual is infected. With a slight abuse of notation, %
	let $X_{\calC_i} \in \{0,1\}, \, i \in [m],$ be the infection status of community $\calC_i$, where $X_{\calC_i} = 1$ iff $\exists i \in \calC_i : X_i = 1$. The indicator function for an event $\calA$ is given by $\mathbbm{1}_\calA$. The entropy of a discrete random variable and the binary entropy function (both in bits) are $H(\cdot)$ and $\mathsf{h}_\mathsf{b}(\cdot)$, respectively. We write $f(x) \prec g(x)$ to denote $f(x) = o(g(x))$, and $f(x) \preceq g(x)$ to denote $f(x) = O(g(x))$. %

	\paragraph{Paper Organization} The rest of this paper is organized as follows. 
	In Section~\ref{sec:background}, we provide background and preliminary results. In Section~\ref{sec:models}, we introduce the stochastic block infection model ($\msf{SBIM}$) and discuss its equivalence to a certain graph-based infection spread mechanism acting upon the stochastic block model. %
	In Section~\ref{sec:algs}, we discuss the main algorithms studied in this paper: binary splitting and our proposed community-aware algorithm. Section~\ref{sec:sbm} gives upper and lower bounds for adaptive group testing over the general SBIM, and Section~\ref{sec:disjoint} provides an in-depth treatment of the disjoint $k$-cliques model, which is a special case of the SBIM. 
	We then extend our algorithms and bounds to the noisy case in Section~\ref{sec:noisy}. Finally, we present the results of our numerical experiments in Section~\ref{sec:simulations}, and conclude in Section~\ref{sec:conclusion}. 
	All omitted proofs are given in the Appendix.

	\section{Background and Preliminary Results }\label{sec:background} 

    \subsection{The Group Testing Problem}
    In the group testing problem, a \emph{test} corresponds to a subset of individuals $\calS \subseteq [n]$. The test outcome is \emph{positive} if $X_i = 1$ for some $i \in \calS$; that is, if at least one member of $\calS$ is infected. Otherwise, the outcome is \emph{negative}. Equivalently, the outcome is a binary variable $Y \in \{0,1\}$ given by a Boolean \texttt{OR} operation over $\calS$:
    \begin{equation}\label{eqn:booleanOR}
        Y = \bigvee_{i \in \calS} X_i.
    \end{equation}
    A group testing algorithm describes how to select subsets $\calS_1,\ldots, \calS_T$ and---given the corresponding test outcomes $Y_1,\ldots, Y_T$---how to generate an estimate $\hat{X}$ of $X$. In \emph{adaptive} schemes, the subsets $\calS_t$ are chosen sequentially and are allowed to depend on the previous test outcomes. %
    In the first part of this paper, we assume that test outcomes are \emph{noiseless} (meaning the algorithm gets to observe $Y$ as given in (\ref{eqn:booleanOR})), and we require \emph{exact recovery} of $X_1,\ldots, X_n$ (i.e., zero error). 
    
    Subsequently, we consider a \emph{noisy} variant of the problem in which the test outcomes are given by
    \begin{equation}\label{eqn:noisy_model}
    Y = \Big(\bigvee_{i \in \calS} X_i \Big) \oplus \xi, 
    \end{equation}
    where $\xi \sim \msf{Bernoulli}(\rho)$ for some $\rho \in (0,\frac{1}{2})$, and $\oplus$ denotes modulo-2 addition. This is the widely-adopted \emph{symmetric noise model} \cite{chan2011non,scarlett2016phase,teo2022noisy}, and it is equivalent to passing each noiseless test outcome through a binary symmetric channel with crossover probability $\rho$. It is assumed that tests are subject to independent noise.  Due to the uncertainty in the test outcomes, we can no longer guarantee exact recovery of $X$. Instead, we seek to ensure a vanishing error probability $P_e \eqDef \Pr(\hat{X} \neq X)$, where the randomness is due to the infection statuses and the noisy test outcomes. 
    
    In our setting, the number of tests $T$ performed by an adaptive scheme is a random variable because it depends on the $X_i$, which are generated by our probabilistic $\msf{SBIM}$ model, as well as the (possibly noisy) test results. Our goal is to characterize the average complexity of adaptive schemes under the $\msf{SBIM}$ by providing both upper and lower bounds on $\EE[T]$.

\subsection{Information-Theoretic Lower Bounds} 
	
	A fundamental result in probabilistic group testing %
	is that \emph{any} adaptive algorithm which is guaranteed to identify all infected members of the population, assuming noiseless test results, requires a number of tests $T$ satisfying 
	\begin{equation}\label{eqn:entropy_LB}
	    \EE[T] \geq H(X_1,\ldots, X_n).
	\end{equation}
	This bound highlights the intimate connection between adaptive group testing and source coding. Indeed, to summarize a discussion from \cite{wolf1985born}, the outcomes of the adaptive tests can be viewed as a binary, variable-length source code for $X$. The lower bound then follows directly from existing results in data compression (e.g., \cite[Eqn.~5.38]{cover2006elements}). Equation (\ref{eqn:entropy_LB}) will serve as the point of departure for the lower bounds on $\EE[T]$ that we derive under the $\msf{SBIM}$ in the noiseless case. The key challenge will be to obtain good approximations to $H(X)$ in the presence of correlations induced by the underlying network.

	For the noisy setting, we prove the following counterpart to \eqref{eqn:entropy_LB}. This lower bound holds for any adaptive scheme and any underlying stochastic infection model, including those with correlations. We provide the proof in Section~\ref{sec:noisy}.
	
    \begin{theorem}\label{thm:noisy_LB}
    Assume $H(X_1,X_2,\ldots, X_n) \to \infty$ as $n \to \infty$. Under the symmetric noise model \eqref{eqn:noisy_model}, any adaptive algorithm achieving $P_e \to 0$ must use an average number of tests lower bounded as 
    \begin{equation}
        \EE[T] \geq \frac{H(X_1,\ldots, X_n)}{I(\rho)},
    \end{equation}
    where $I(\rho) = 1-\msf{h}_\msf{b}(\rho) = 1 -  \rho \log_2 \frac{1}{\rho} - (1-\rho) \log\frac{1}{1-\rho}$ is the capacity of the binary symmetric channel with crossover probability $\rho \in (0, \frac{1}{2})$.
    \end{theorem}
    
    Note that our bound recovers the noiseless lower bound \eqref{eqn:entropy_LB} when $\rho=0$. 
    Moreover, in the special case of the combinatorial prior model where the number of infections $d$ is fixed and the set of infected members is uniformly distributed over the $\binom{n}{d}$ possibilities, our bound reduces to
    $\EE[T] \geq \frac{\log \binom{n}{d}}{I(\rho)}$.
    A version of this bound appears in \cite{teo2022noisy}, which does not prove it directly but argues it can be shown using an existing result from \cite{baldassini2013capacity} along with the variable-length coding capacity of the binary symmetric channel. In Section~\ref{sec:noisy}, we provide a stand-alone proof of the more general lower bound in Theorem~\ref{thm:noisy_LB} which encompasses all adaptive schemes and probabilistic infection models with symmetric testing noise. %
    Though we will primarily focus on the implications of Theorem~\ref{thm:noisy_LB} in the context of the $\msf{SBIM}$, we again emphasize that this result is independent of any particular infection model and thus can be of interest in its own right.

\section{ Stochastic Block Infection Model (SBIM)}\label{sec:models} 
In this section, we introduce the \emph{stochastic block infection model} (SBIM), which extends the traditional i.i.d. group testing model to a community-oriented setting. Here, individuals infect their fellow community members with a higher probability than those in other communities, giving rise to strongly correlated clusters of infections. We then describe a special case of the SBIM---the \emph{disjoint $k$-cliques model}---in which the communities are ``disconnected.'' Finally, we discuss the relationship between the SBIM and the stochastic block model (SBM), and conclude the section with some practical considerations. 

\subsection{General SBIM} 
Assume we are given a partition of the population of size $n$ into $m \eqDef n/k$ communities $\calC_1,\ldots, \calC_m$ of size $|\calC_i| = k, \, \forall i \in [m]$. The SBIM comprises the following two stages (each executed once):

	\begin{enumerate}
		\item \textbf{Seed Selection: } Individuals in the population are infected independently with probability $p \in (0,1]$. These initial infected members are called the \textit{seeds}. They model the introduction of the disease into the population via some external entity (e.g., a traveler carrying the disease into a country).  
		\item \textbf{Neighbor Infection: } Every seed infects its neighbors within the same community independently with probability $q_1 \in [0,1]$ and those outside its community independently with probability $q_2 \in [0,1]$, where $q_1 > q_2$. 
		This models how the disease spreads through the population via interactions between carriers and nearby individuals. Members of the same community are more likely to interact with each other within a given time frame (e.g., by interacting socially or professionally, or by being in the same physical space, e.g., a supermarket or a restaurant) and therefore more likely to infect each other than members of different communities.
	\end{enumerate}
	
We denote this model by $\mathsf{SBIM}(n,k,p,q_1,q_2)$. Note that $\msf{SBIM}(n,k,p,0,0)$, for any value of $k$, is equivalent to the i.i.d. group testing model with prior probability $p$. We assume the communities are known to the group testing algorithms in advance, but that nothing more is known about the specific interactions  between individuals.  %

The SBIM can be viewed as a model for the initial spread of an epidemic. It is motivated in part by diseases such as COVID-19, which are introduced into a population from an external source and subsequently transmitted between individuals in close contact. We also believe the SBIM can be a natural model for other application areas where group testing has played a role. For example, in the context of coding for multiple access sensor networks \cite{inan2018energy, inan2019group, inan2019sparse} it can capture the fact that sensors in close proximity can have correlated activity patterns and measurements.

\subsection{Special Case: Disjoint $k$-Cliques Model} After analyzing the $\mathsf{SBIM}$ in full generality in Section~\ref{sec:sbm}, we thoroughly investigate the special case of
$\mathsf{SBIM}(n,k,p,q,0)$, which we call the \textit{disjoint $k$-cliques model}, in Section~\ref{sec:disjoint}. Here, we have $m \eqDef n/k$ communities of size $k$, with seed selection probability $p$, intra-community transmission rate $q$, and an inter-community transmission rate of zero. Thus, the communities can be treated as independent, as no transmissions between communities are possible. We note that in this special case our model becomes similar (but not equivalent) to the disjoint families model introduced in \cite{nikolopoulos2020community}. We comment further on this in Section~\ref{sec:disjoint}. Figure~\ref{fig:network_infection} illustrates the $\mathsf{SBIM}(n,k,p,q_1,q_2)$ and contrasts the disjoint $k$-cliques model ($q_2 = 0$) with the general SBIM ($q_2 > 0$).

	\begin{figure}[t]
    \centering
    \begin{subfigure}[t]{0.35\textwidth}
    \includegraphics[width=\textwidth]{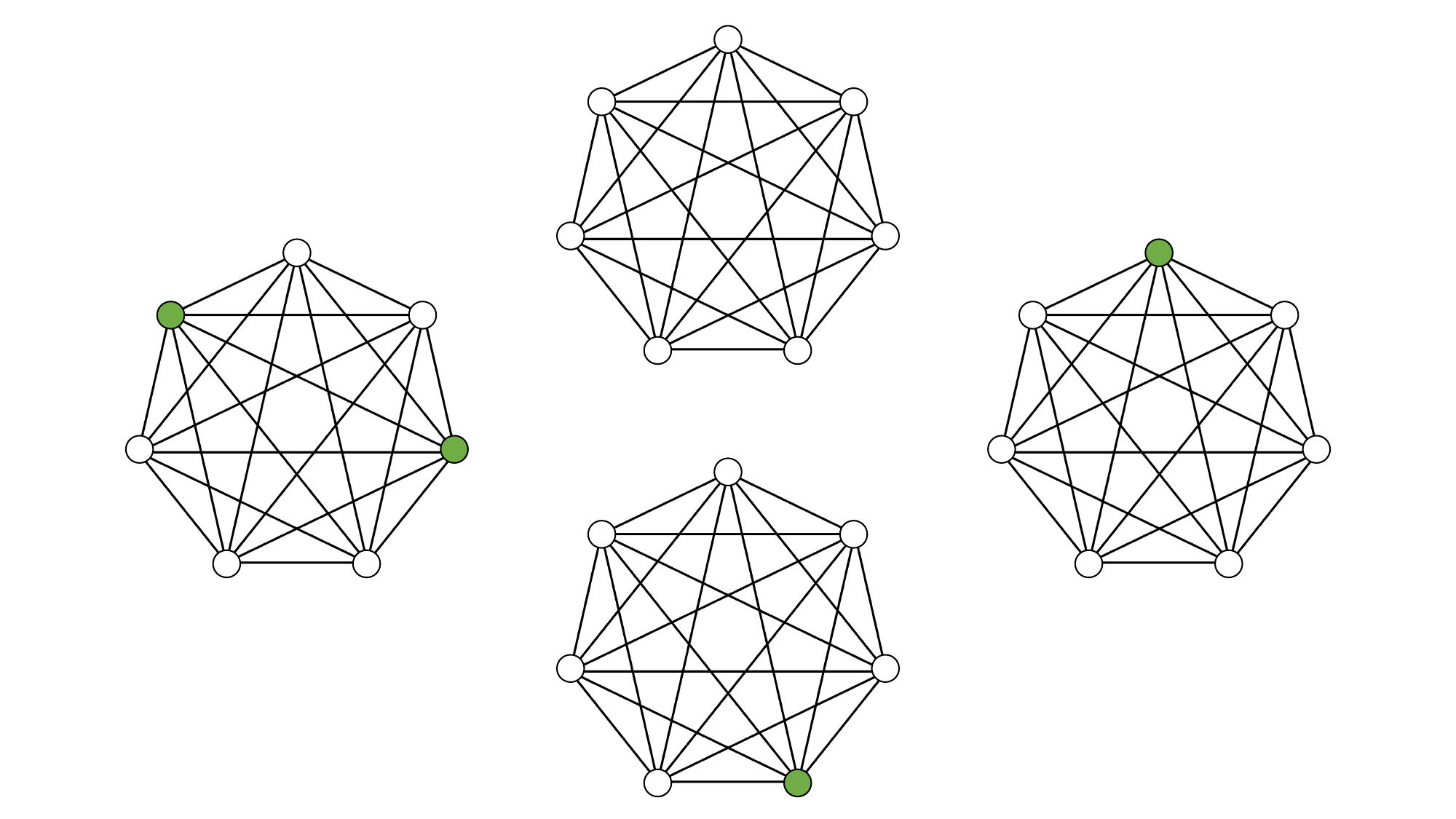}
    \caption{Seed selection stage}
    \end{subfigure}\\
    \begin{subfigure}[t]{0.35\textwidth}
    \includegraphics[width=\textwidth]{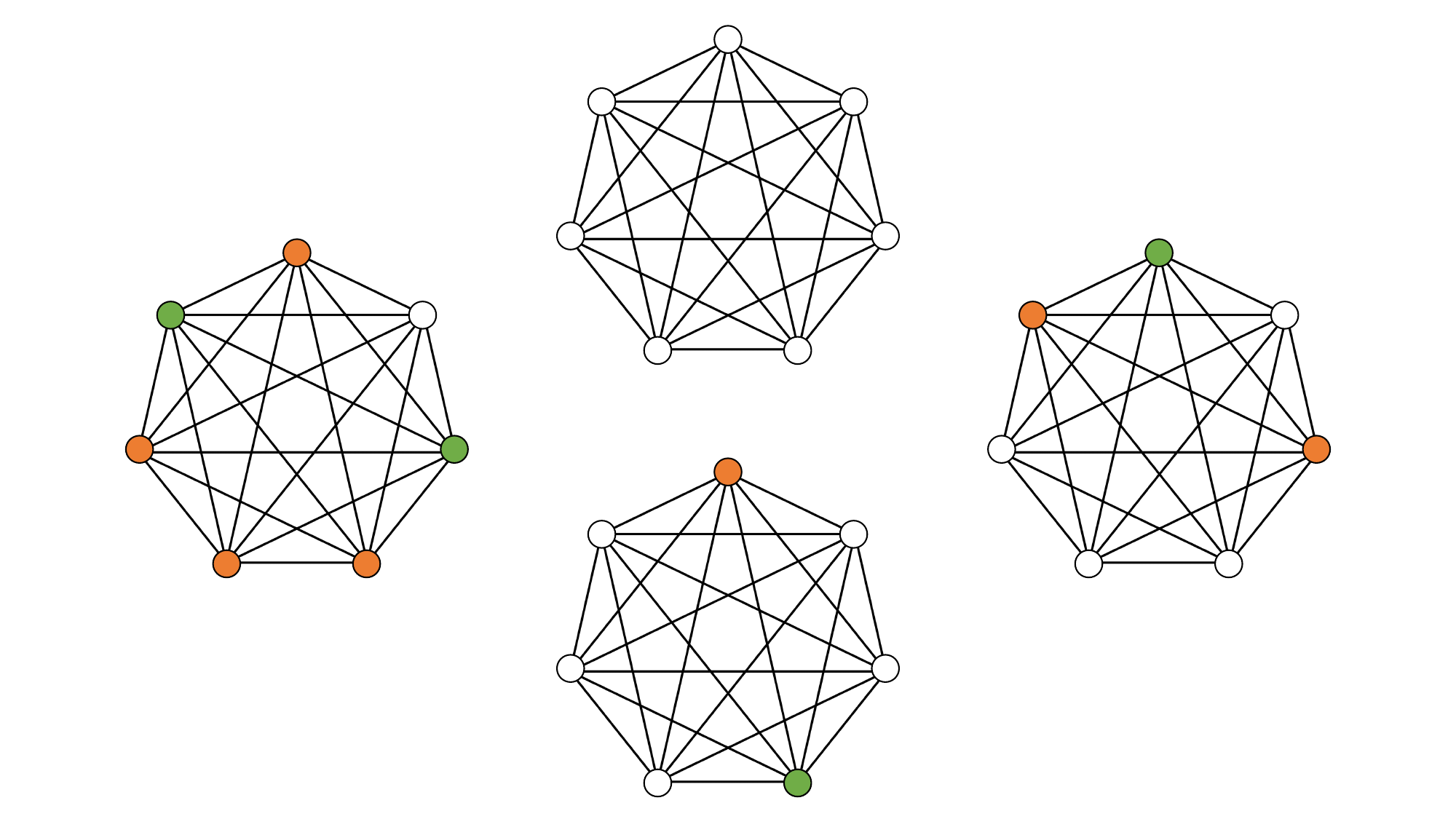}
    \caption{Neighbor infection with $q_2 = 0$ (the disjoint $k$-cliques model). Individuals cannot be infected by seeds outside their own community.} 
    \end{subfigure}\hspace{8em}
    \begin{subfigure}[t]{0.35\textwidth}
    \includegraphics[width=\textwidth]{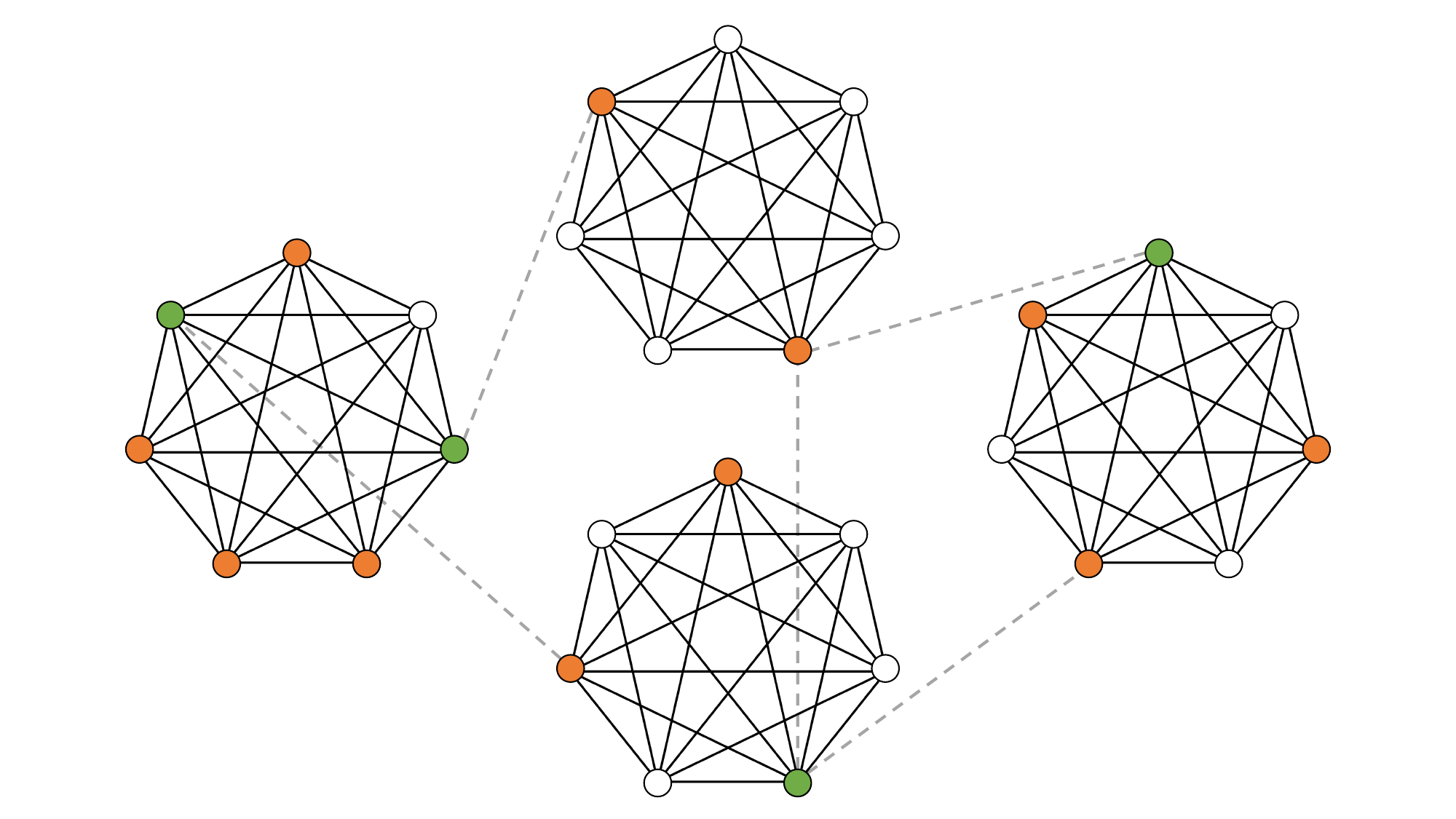}
    \caption{Neighbor infection with $q_2 > 0$ (the general SBIM). Any individual can be infected by any seed, even those in external communities.}
    \end{subfigure} 
    \caption{Illustration of $\mathsf{SBIM}(n,k,p,q_1,q_2)$. In this example, there are $m = 4$ communities of size $k = 7$. Seeds are colored green, and individuals infected by seeds during the neighbor infection stage are colored orange.} 
    \label{fig:network_infection}
\end{figure}
	
	\subsection{Relationship to the Stochastic Block Model (SBM)}\label{subsec:SBM_equivalence} 
	The SBIM is equivalent to a certain graph-based infection spread model operating upon the \emph{stochastic block model} (SBM) \cite{holland1983stochastic}---a well-known random graph model with the tendency to produce graphs with community structure. The standard SBM produces a random undirected graph $\calG = (\calV,\calE)$ as follows. (In our context, the vertices $\calV$ represent members of the population, and the edges $\calE$ can be thought of as representing an interaction (e.g., a social or professional interaction, being in proximity, etc.) between the two members of the population in a time frame of interest.) First, it is assumed the $n$ vertices are partitioned into $m$ communities, $\calC_1,\ldots, \calC_m$, where $\bigcup_{i \in [m]} \calC_i = \calV$ and $\calC_i \cap \calC_j = \emptyset, \, \forall i \neq j$. In addition, we are given a symmetric matrix $\mathbf{P} \in \RR^{m \times m}$ of edge probabilities. The graph is then generated by first initializing $\calE = \emptyset$, then adding an edge between each pair of vertices $u \in \calC_i, \, v \in \calC_j, \, u \neq v,$ with probability $\mathbf{P}_{ij}$. 
	
    Consider a special case of the SBM in which the communities are all of size $k$ (where $k$ is a factor of $n$), the edge probabilities within each community are constant ($p_1$), and the edge probabilities between communities are also constant ($p_2$, where $p_2 < p_1$, which models the assumption that members of the same community are more likely to have an interaction). That is, the diagonal entries of $\mathbf{P}$ are $p_1$, and the off-diagonal entries are $p_2$. %
    Additionally, consider the following probabilistic infection model with parameters $p,q \in [0,1]$ which operates upon an arbitrary graph $\calG = (\calV, \calE)$. First, the vertices are infected independently with probability $p$, producing the seeds $\calV_s \subseteq \calV$. Next, every seed $v \in \calV_s$ infects its neighbors $\calN(v) = \{u \in \calV \, : \, \{u,v\} \in \calE\}$ independently with probability $q$. This models the fact that if two members have an interaction, the disease is transmitted between them with a certain probability $q$. \footnote{This infection model forms the ``first time step'' of the independent cascade model \cite{kempe2003maximizing} from the study of influence maximization in social networks.} Note that this infection model reduces to the i.i.d. group testing model with prior $p$ when (i) $q = 0$, or (ii) $\calG$ is the empty graph ($\calE = \emptyset$). %
    Moreover, by setting $q_1 = p_1 \cdot q$ and $q_2 = p_2 \cdot q$, we see that the $\mathsf{SBIM}(n,k,p,q_1,q_2)$ is equivalent to this infection model operating upon the SBM. %

\subsection{Practical Considerations} 
	The ``communities'' within the SBIM can represent populations at different scales: counties, cities, schools, companies, etc. In practice, the specific values of $p,q$ can be tailored to the disease in question (for example, by using contact tracing to estimate the infectiousness of the disease). 
	Lastly, when the communities are not known in advance, one might first estimate the network from data (e.g., contact tracing, mobile phone, or social network data), then run a graph clustering algorithm to identify communities in the network. %
	
	At the same time, we acknowledge the practical limitations of the SBIM. First, the symmetry of the model (e.g., the assumptions that every community has the same probability of containing a seed and that a given individual can be infected by a seed from any community) does not capture the reality that the transmissibility of a disease can vary from person to person depending on their habits (e.g., whether they practice social distancing or mask wearing). However, we still believe the SBIM is an important and natural ``first-order'' extension of the traditional i.i.d. group testing model (which has been studied for decades) to models of greater complexity and practical relevance, while still being analytically tractable. We note that some of the aforementioned issues can be incorporated through the study of the general graph infection model we introduced in Section~\ref{subsec:SBM_equivalence}, by assuming the matrix $\mathbf{P}$ has a more general structure, e.g., by allowing different edge probabilities within and/or between different blocks, and/or different block sizes. Studying the group testing problem in these more general settings is an exciting direction to pursue in future work.

	\section{Algorithms}\label{sec:algs} 

\subsection{Binary Splitting Algorithm}
    	
    Most adaptive group testing algorithms are based on the idea, first introduced by Sobel and Groll \cite{sobel1959group}, of recursively splitting the population until all infected members are  found. The most fundamental adaptive procedure is \textit{binary splitting}, which finds a single infected member at a time by repeatedly halving the population. %
    It works even when the number of infected members $d$ is unknown \cite{sobel1966binomial}, and is most effective in the sparse regime, $d = \Theta(n^\beta)$, where $\beta \in [0,1)$.  %
    We make extensive use of the following performance guarantee throughout this paper:\footnote{ It is well-known that Lemma~\ref{lem:binary_split} can be improved via Hwang's \emph{generalized binary splitting} algorithm \cite{hwang1972method} or Allemann's \emph{split and overlap} algorithm \cite{allemann2013efficient}.  However, in contrast to binary splitting, these methods require the number of infected individuals (or an upper bound on this quantity) to be known \emph{a priori}.
    }  %
    	
    \begin{lemma}\label{lem:binary_split} 
    	In a population of size $n$ with $d$ infected members---where $d$ is unknown---%
    	the binary splitting algorithm is guaranteed to identify all infected members using at most $d \ceil{\log_2 n} + d + 1 \leq d \log_2 n + 2d + 1$ tests.
    \end{lemma}
    \begin{proof}
    The first step of binary splitting is to perform a single test on the entire population to check for the presence of an infected member. If the test is positive, an infected member is identified in a recursive fashion using at most $\ceil{\log_2 n}$ adaptive tests (see \cite{baldassini2013capacity},  \cite[p.24-25]{du2000combinatorial}, or \cite[Theorem~1.2]{aldridge2019group} for details of the proof). The infected individual is then removed from the population and the aforementioned steps are repeated until either no individuals remain, or a negative test is obtained in the first step. It is straightforward to see that $d+1$ tests are performed due to the first step (once per infected member, and again when no infections remain), and $d \ceil{\log_2 n}$ total tests are used to recursively identify all $d$ infected members.
    \end{proof}
    
     We treat binary splitting as the baseline algorithm in this paper due to its simplicity and its role as a key subroutine in many other adaptive procedures. 
    
    \subsection {Community-Aware Algorithm} 
    As an alternative to standard adaptive procedures such as binary splitting, we consider a simple two-stage scheme which leverages the community structure of the graph. Our scheme first treats the communities as ``meta-individuals'' by mixing the samples within each community and applying binary splitting to quickly identify those with at least one infected member. Subsequently, we run binary splitting again---this time within each infected community---to identify the infected individuals. Note that this procedure will recover the infection statuses of all members of the population with zero error, which follows from the fact that binary splitting achieves exact recovery.

\vspace{1em}  
    \noindent
	\fbox{ \parbox{0.97\textwidth}{
	\textbf{Adaptive Community-Aware Algorithm} 
	\begin{enumerate}
		\item Mix the samples within each community.
		\item Perform binary splitting on the mixed samples to determine which communities contain at least one infected member. 
		\item For each positive test from Step 2, perform binary splitting within the corresponding community to identify the infected members.
	\end{enumerate}
	}}
	\vspace{1em}

	Under what circumstances would we expect the community-aware algorithm to outperform binary splitting? Suppose the underlying model is $\mathsf{SBIM}(n,k,p,q_1,q_2)$. If the seed selection probability $p$ is small, then we expect only a few of the $m \eqDef n/k$ communities to contain a seed. Thus, after the neighbor infection stage, several of the communities are likely to contain no infected members at all, especially if $q_2$ is small. In Step~2 of the community-aware algorithm, we can efficiently rule out these uninfected communities from consideration. In Step~3, we need only perform group testing within each of the remaining communities (which contain at least one infected member). In contrast, the binary splitting algorithm ignores the community structure (specifically, the fact that entire communities are likely to be uninfected), and is therefore unlikely to enjoy the same benefits as the community-aware algorithm under these circumstances. We will rigorously verify this intuition in the upcoming sections.
    
    \section{Bounds for the SBIM}\label{sec:sbm} 
In this section, we derive general lower and upper bounds on the average complexity of adaptive group testing over the $\mathsf{SBIM}(n,k,p,q_1,q_2)$. %
As we saw in the previous section, the community-aware algorithm is a simple extension of the binary splitting algorithm to a community-oriented setting. From a technical perspective, our main contribution is a careful evaluation of the performance of these schemes for the $\msf{SBIM}$ model as well as  the system entropy $H(X)$, which is required to obtain meaningful lower bounds. We start with the lower bound. %

\subsection{Information-Theoretic Lower Bound}\label{subsec:sbm_LB}
Recall from \eqref{eqn:entropy_LB} that $\EE[T] \geq H(X)$ for any adaptive group testing algorithm which exactly identifies the infected individuals using $T$ tests. The following lemma gives both a general lower bound on $H(X)$ as well as an easier-to-compute bound in terms of two independent binomial random variables. The proof is in Appendix~\ref{app:cond_entropy_LB}.

\begin{lemma}\label{lemma:sbm_LB}
Let $X_1,...,X_n$ be the infection statuses generated from $\msf{SBIM}(n, k, p, q_1, q_2)$, as defined in Section~\ref{sec:models}, and let $S_i$ be the indicator variable of whether the $i^\text{th}$ individual is a seed. Then the number of tests $T$ required to identify the infected individuals is lower bounded as
\begin{align}
    \EE[T] \geq H(X_1,...,X_n) &\geq n \cdot I(X_1;S_1) + H(X_1,\ldots,X_n\,|\,S_1,\ldots,S_n) \label{eqn:full_LB} \\
    &\geq m \cdot \EE_{Z, Z'} \left[\left( k - Z\right)\cdot \mathsf{h}_\mathsf{b}\lp 1-(1-q_1)^Z\lp 1- q_2 \rp^{Z'}\rp\right], \label{eqn:cond_entropy_LB}
\end{align}
where $Z \sim \mathsf{Binom}\left( k, p\right)$ and $Z'\sim \msf{Binom}(n-k, p)$ are independent.
\end{lemma}

By leveraging the concentration of $Z$ and $Z'$ around their means, we obtain our first main result, which characterizes the asymptotic behavior of \eqref{eqn:cond_entropy_LB}. The proof is in Appendix~\ref{app:sbm_LB}.

\begin{theorem}[SBIM Lower Bound]\label{thm:sbm_LB}
Assume
\begin{enumerate}
    \item $n\cdot p\cdot q_2 \preceq 1$,
    \item $n\cdot p \succeq 1$,
    \item $k\cdot p\cdot q_1\preceq 1$,
    \item $q_1 \leq \frac{1}{\sqrt{2k\lp \log\lp \frac{1}{kp} \rp+1  \rp}}$.
\end{enumerate}
The number of tests $T$ needed to recover all infected members over $\mathsf{SBIM}(n,k,p,q_1,q_2)$ is lower bounded as
$$
\EE[T] \succeq m^2\cdot k^2 \cdot p \cdot q_2 \cdot \log\lp \frac{1}{n\cdot p\cdot q_2} \rp + m\cdot k^2 \cdot p\cdot q_1 \cdot  \log\lp \frac{1}{q_1 + n\cdot p \cdot q_2} \rp.
$$
\end{theorem} 

\begin{remark}

The upper bounds on $p, q_1$ and $q_2$ in Theorem~\ref{thm:sbm_LB} allow us to evaluate the lower bound in Lemma~\ref{lemma:sbm_LB} in a regime where the infected population is sparse enough. This is the relevant regime since group testing is known to improve upon individual testing when infections are sparse. However, the specific upper bounds we impose may be artifacts of our lower bounding technique and could potentially be loosened. 
\end{remark}

As we will see in Section~\ref{subsec:k-clique_LB}, a secondary lower bound under $\msf{SBIM}(n,k,p,q,0,0)$ (i.e., the disjoint $k$-cliques model) is given by $H(X_{\calC_1},\ldots, X_{\calC_m}) = \sum_{i \in [m]} H(X_{\calC_i})$, which leverages the fact that the community-level infection statuses $\lbp X_{\calC_1},...,X_{\calC_m} \rbp$ are mutually independent in this setting. This bound turns out to dominate when $kp \preceq m^{-\beta}$ for some fixed $\beta \in (0,1)$. It is difficult to obtain an analogous lower bound under the general SBIM since the $\lbp X_{\calC_1},...,X_{\calC_m} \rbp$ are no longer mutually independent when $q_2 > 0$. Therefore, we suspect that the lower bound given in Theorem~\ref{thm:sbm_LB} is not tight when $kp$ is small. Obtaining a tighter bound in this regime is an open problem.

\subsection{Algorithm Analysis}
To analyze binary splitting and the community-aware algorithm over the SBIM, we begin by characterizing the marginal probability that a given individual will be infected. The proof is in Appendix~\ref{app:sbm_marginal}. %

\begin{lemma}\label{lem:sbm_marginal} 
The marginal probability of infection for every individual under $\mathsf{SBIM}(n,k,p,q_1,q_2)$ is given by
\[\PP(X_v = 1) = 1 - (1-p) \cdot (1-p \cdot q_1)^{k-1} \cdot (1-p \cdot q_2)^{n-k}.\]
\end{lemma}

\subsubsection{Binary Splitting} 

The following result bounds the expected number of tests used by the binary splitting algorithm under the SBIM. 

\begin{theorem}[Binary Splitting Bound]\label{thm:sbm_binary_split} 
Under $\mathsf{SBIM}(n,k,p,q_1,q_2)$, the binary splitting algorithm identifies all infected individuals using $T$ tests, where 
\[\EE[T] \leq n \cdot (\log_2 n + 2) \cdot \Big(1 - (1-p) \cdot (1-p \cdot q_1)^{k-1} \cdot (1-p \cdot q_2)^{n-k}\Big) + 1.\]
\end{theorem}

\begin{proof}
Let $K$ be the number of infected individuals. Then 
		\[\EE[K] = \EE\Big[\Sum{i=1}{n} X_i\Big] = \Sum{i=1}{n} \PP(X_i = 1) = n \cdot \Big(1 - (1-p) \cdot (1-p \cdot q_1)^{k-1} \cdot (1-p \cdot q_2)^{n-k}\Big)\]
	where %
	the last equality follows from Lemma~\ref{lem:sbm_marginal}. 
	Invoking Lemma~\ref{lem:binary_split} yields the result. 
\end{proof}

\begin{corollary}\label{cor:sbm_binary_split}
Under $\mathsf{SBIM}(n,k,p,q_1,q_2)$, the average complexity of binary splitting satisfies 
\[\EE[T] \preceq m\cdot k^2 \cdot p \cdot (\log m + \log k) \cdot \Big(\frac{1}{k} + q_1 + m \cdot q_2 + m\cdot k \cdot p^2 \cdot q_1 \cdot q_2\Big). \]
\end{corollary}
\begin{proof} 
Using the fact that $(1+x)^k \geq 1 + kx$ for $x \geq -1, \, k \geq 1$, we have 
\begin{align}
    \EE[T] &\preceq n \cdot \log n \cdot \Big(1 - (1-p)(1-k\cdot p \cdot q_1) \cdot (1-(n-k)\cdot p \cdot q_2)\Big) \nonumber \\
    &\leq n \cdot \log n \cdot \Big((n-k)\cdot p \cdot q_2 + k \cdot p \cdot q_1 + p + k\cdot (n-k)\cdot p^3 \cdot q_1 \cdot q_2\Big) \nonumber \\
    &\leq  m\cdot k^2 \cdot p \cdot (\log m + \log k) \cdot \Big(\frac{1}{k} + q_1 + m \cdot q_2 + m\cdot k \cdot p^2 \cdot q_1 \cdot q_2\Big).\label{eqn:sbim_bs_asymptotic}
\end{align} 
\end{proof}

\subsubsection{Community-Aware Algorithm} 
First, we provide a lemma needed to prove the upper bound for the community-aware algorithm. The proof is in Appendix~\ref{app:sbm_clique_marginal}.  %

\begin{lemma}\label{lem:sbm_clique_marginal} 
Let $X_{\calC_i} \in \{0,1\}, \, i \in [m],$ be the infection status of community $\calC_i$, where $X_{\calC_i} = 1$ iff there exists at least one infected member in $\calC_i$.
Under $\mathsf{SBIM}(n,k,p,q_1,q_2)$,
\[\PP(X_{\calC_1} = 1) = 1 - (1-p)^k \cdot \Bigg(1 - p\cdot \Big(1-(1-q_2)^k\Big)\Bigg)^{n-k}.\]
\end{lemma}

In Theorem~\ref{thm:sbm_graphaware} below (which is proved in Appendix~\ref{app:sbm_graphaware}), the two terms in the sum correspond, respectively, to the expected number of tests in Steps 2 and 3 of the community-aware algorithm. 

\begin{theorem}[Community-Aware Bound]\label{thm:sbm_graphaware}
Under $\mathsf{SBIM}(n,k,p,q_1,q_2)$, the community-aware algorithm identifies all infected individuals using $T$ tests, where 
\begin{align*}
\EE[T] \leq \frac{n}{k} \cdot &\Big(\log_2(n/k) + 3\Big) \cdot \Bigg(1 - (1-p)^k \cdot \Bigg(1 - p\cdot \Big(1-(1-q_2)^k\Big)\Bigg)^{n-k}\Bigg) \\
&+ n \cdot \Big(\log_2 k + 2\Big) \cdot \Big(1 - (1-p) \cdot (1-p \cdot q_1)^{k-1} \cdot (1-p \cdot q_2)^{n-k}\Big) + 1.
\end{align*} 
\end{theorem}

\begin{corollary}\label{cor:sbm_graphaware} 
Under $\mathsf{SBIM}(n,k,p,q_1,q_2)$, the average complexity of the community-aware algorithm satisfies 
\begin{equation}\label{eqn:sbim_ga_asymptotic}
\EE[T] \preceq m  \cdot k \cdot p \cdot \log m \cdot \Big(1 +  m \cdot k \cdot q_2\Big)  + m \cdot k^2 \cdot p\cdot \log k \cdot \Big(\frac{1}{k} + q_1 + m \cdot q_2 + m \cdot k \cdot p^2 \cdot q_1 \cdot q_2\Big)
\end{equation} 
\end{corollary}
\begin{proof} 
Let $T_1$ and $T_2$ be the first and second terms in the Theorem~\ref{thm:sbm_graphaware} bound, respectively. 
Using the fact that $\lp 1-q_2\rp^k \geq 1-kq_2$, we have
$$ 1-p\cdot \lp1-(1-q_2)^k\rp \geq 1-p\cdot k \cdot q_2, $$
so 
\begin{align*}
    \EE[T_1] 
    & \preceq m\log m\cdot\lp 1-(1-p)^k \cdot\lp 1-p\lp1-(1-q_2)^k\rp\rp^{n-k}\rp \\
    & \preceq m\log m \cdot\lp 1- (1-p)^k \cdot \lp 1 - p\cdot k\cdot q_2 \rp^{n-k} \rp\\
    & \preceq m\log m \cdot \lp 1-(1-k\cdot p) \cdot (1- \lp n-k \rp \cdot p \cdot k \cdot q_2) \rp \\
    & \preceq m\log m \cdot\lp k \cdot p + n\cdot p \cdot k \cdot q_2\rp.
\end{align*}
We can then bound $\EE[T_2]$ by following the previous asymptotic analysis for binary splitting: 
\begin{align*}
    \EE[T_2] \preceq m \cdot k^2 \cdot p \cdot \log k \cdot \Big(\frac{1}{k} + q_1 + m \cdot q_2 + m \cdot k \cdot p^2 \cdot q_1 \cdot q_2\Big).
\end{align*}
\end{proof}

\subsection{Discussion}\label{subsec:sbm_discussion} 
 
Comparing \eqref{eqn:sbim_bs_asymptotic} and \eqref{eqn:sbim_ga_asymptotic} term-by-term, we see that the binary splitting bound has an extra additive factor of $mk^2pq_1\log m (1 + mkp^2 q_2)$ compared to the community-aware bound, implying that the community-aware algorithm is never worse (order-wise) than binary splitting. Furthermore, one can verify that when $q_1 \preceq 1/k$ and  $q_2 = 0$ (which includes the i.i.d. setting, where community structure has no bearing on the infection spread), the bounds are order-wise equivalent. %
This supports our intuition that knowledge of the community structure may not help when $q_1$ and $q_2$ are small, as the infection statuses of the individuals are ``mostly independent'' in this regime. 

In other regimes, the community-aware algorithm is asymptotically strictly better than binary splitting. %
The main takeaway from the following corollary is that the community-aware algorithm can potentially improve upon binary splitting when there are several moderately sized communities in the network, and the transmission rate within each community is significant. 

\begin{corollary}\label{cor:sbm_improvement}
If $\log m \succ \log k$, $kq_1 \succ 1$, and $1 \succeq m k q_2 $, then the community-aware algorithm's average complexity is asymptotically strictly better than binary splitting's by a factor of $\min\Big\{kq_1, \, \frac{\log m}{\log k}\Big\}$. 

If $\log m \succ \log k$, $kq_1 \succ 1$, $mkq_2 \succeq 1$, and $mkq_2 \prec kq_1 \preceq \frac{1}{p^2}$, then the improvement is a factor of $\min\Big\{\frac{q_1}{m\cdot q_2}, \frac{\log m }{\log k}\Big\}$.

\end{corollary}

\begin{proof} 
    Suppose $\log m \succ \log k$, $kq_1 \succ 1$, and $1 \succeq m k q_2 $.
    Binary splitting's average complexity \eqref{eqn:sbim_bs_asymptotic} becomes 
    $$m \cdot k^2 \cdot p \cdot q_1 \cdot \log m$$
    whereas the community-aware algorithm's average complexity is
    $$\max\Big\{ \underbrace{m \cdot k \cdot p \cdot \log m}_{\text{(a)}}, \quad \underbrace{m \cdot k^2 \cdot p \cdot q_1 \cdot \log k}_{\text{(b)}} \Big\}.$$
    Both (a) and (b) are strictly smaller than the binary splitting bound. We see that (a) saves a factor of $kq_1 \succ 1$, while (b) saves a factor of $\frac{\log m}{\log k} \succ 1$. Thus, the overall improvement is a factor of $\min\Big\{kq_1, \, \frac{\log m}{\log k}\Big\}$.
    
    Next, suppose $\log m \succ \log k$, $kq_1 \succ 1$, $mkq_2 \succeq 1$, and $mkq_2 \prec kq_1 \preceq \frac{1}{p^2}$. Binary splitting's average complexity is $m \cdot k^2 \cdot p \cdot q_1 \cdot \log m$ (same as before), and the community-aware algorithm's complexity becomes 
$$\max\Big\{m^2 \cdot k^2 \cdot p \cdot q_2 \cdot \log m, \quad  m\cdot k^2 \cdot p \cdot q_1 \cdot \log k \Big\},$$
which represents an improvement over binary splitting by a factor of $\min\Big\{\frac{q_1}{m\cdot q_2}, \frac{\log m }{\log k}\Big\} \succ 1$.
\end{proof}

In general, the lower bound given in Theorem~\ref{thm:sbm_LB} exhibits a gap to the upper bounds in Theorem~\ref{thm:sbm_binary_split} and Theorem~\ref{thm:sbm_graphaware}. However, we will see in the next section that this gap can be eliminated under certain assumptions.  
    
    \section{Bounds for the Disjoint $k$-Cliques Model}\label{sec:disjoint}
Having studied the SBIM in full generality, we now focus on the special case of $\mathsf{SBIM}(n,k,p,q,0)$. Here, the transmission rate within a community is $q \in [0,1]$, and no transmissions are possible between communities. This simplifying assumption allows us to obtain a tighter lower bound than in the general SBIM, and to further show that the community-aware algorithm is order-optimal in certain parameter regimes. 

This setting is conceptually similar to the disjoint families model from \cite{nikolopoulos2020community}. However, in their model, each member of an ``infected family'' is infected independently with a fixed probability, whereas the infection rate within a given community in our model depends on the number of seeds in the community, which in turn depends (probabilistically) on the size of the community. This models the realistic scenario where a larger  community has a larger probability of being ``infected,'' i.e., having some infected members. In addition, the state of a given member of an infected community is not independent of the states of the other members; an individual has a higher probability of being infected if there are more infected members in their community. This property of our model makes the derivation of lower bounds and the analysis of group testing schemes more intricate. %

\subsection{Information-Theoretic Lower Bound}\label{subsec:k-clique_LB}

We obtain the following lower bounds for adaptive group testing over the the disjoint $k$-cliques model.

\begin{lemma}\label{lemma:k-clique_LB}
Under the disjoint $k$-cliques model, the number of tests $T$ required to identify the infected individuals is lower bounded as
$$ \EE[T] \geq H(X_1,...,X_n) \geq m\cdot\EE_Z \left[\left( k - Z\right)\cdot \mathsf{h}_\mathsf{b}\left( 1-(1-q)^Z\right)\right], $$
where $Z \sim \mathsf{Binom}\left( k, p\right).$
\end{lemma}
\begin{proof}
Direct consequence of Lemma~\ref{lemma:sbm_LB}.
\end{proof}

Next is a technical lemma which characterizes the asymptotic behavior of Lemma~\ref{lemma:k-clique_LB} by leveraging the concentration of $Z$ around its mean, using similar techniques as in Theorem~\ref{thm:sbm_LB}.
The proof is in Appendix~\ref{app:k-clique_LB}.

\begin{lemma}\label{lem:k-clique_LB} 
Let $Z \sim \mathsf{Binom}\left( k, p\right)$ and assume $kp \preceq 1$ and $q \preceq \frac{1}{\sqrt{k}\cdot\sqrt{\log\left( \frac{1}{k\cdot p} \right)}}$. Then 
$$ \EE_Z \left[\left( k - Z\right)\cdot \mathsf{h}_\mathsf{b}\left( 1-(1-q)^Z\right)\right] \succeq k^2\cdot p\cdot q \cdot \lp\log k + \log\log\lp \frac{1}{k\cdot p} \rp\rp. $$
\end{lemma}

Upon combining Lemma~\ref{lemma:k-clique_LB} and Lemma~\ref{lem:k-clique_LB}, we see that the number of tests $T$ needed to recover all infected members in the disjoint $k$-cliques graph (in the specified parameter regime) is lower bounded as 
\begin{equation}\label{eqn:clique_entropy1}
\EE[T] \succeq m\cdot k^2\cdot p\cdot q \cdot \lp\log k + \log\log\lp \frac{1}{k p} \rp\rp.
\end{equation} 
Recall that $X_{\calC_i}$ is the indicator variable of whether community $\calC_i$ contains at least one infected member. A different lower bound is given by
\begin{align} 
\EE[T] \geq H(X_1,\ldots, X_n) \overset{\text{(a)}}{\geq} H(X_{\calC_1},\ldots, X_{\calC_m}) = m\cdot \mathsf{h}_\mathsf{b}\Big(1-(1-p)^k\Big) \overset{\text{(b)}}{\succeq} m \cdot k\cdot p \cdot \log_2(1/kp) \label{eqn:clique_entropy2}
\end{align}
where (a) uses the fact that $X_{\calC_1},\ldots, X_{\calC_m}$ are a function of $X_1,\ldots, X_n$, and (b) uses the fact that  $\mathsf{h}_\mathsf{b}\Big(1-(1-p)^k\Big) \succeq k\cdot p \cdot \log_2(1/kp)$ since $kp \preceq 1$.

In the following theorem, we summarize the refined lower bound obtained by combining \eqref{eqn:clique_entropy1} and \eqref{eqn:clique_entropy2}: 

\begin{theorem}[Disjoint $k$-Cliques Lower Bound]\label{thm:k-clique_LB} 
Assume $kp \preceq 1$ and $q\preceq \frac{1}{\sqrt{k\log\lp \frac{1}{kp} \rp}}$. Then under the disjoint $k$-cliques model, the expected number of tests required to identify the infected individuals is lower bounded as 
$$ \EE[T] \succeq \max\lbp m\cdot k^2\cdot p\cdot q \cdot \lp\log k + \log\log\lp \frac{1}{k \cdot p} \rp\rp, \quad  m \cdot k \cdot p \cdot \log\Big(\frac{1}{k \cdot p}\Big), \quad  1\rbp. $$
\end{theorem}

Recall that $q=0$ corresponds to i.i.d. group testing, in which case \eqref{eqn:entropy_LB} gives the lower bound $\EE[T] \geq n\cdot \mathsf{h}_\mathsf{b}(p) \geq np\log(1/p)$. On the other hand, substituting $q=0$ into Theorem~\ref{thm:k-clique_LB} yields $np\log(1/kp)$, which differs from the i.i.d. case by an additive factor of $np\log(1/k)$. In this special case, our bound can be seen as slightly suboptimal.
However, observe that when $q=0$, the disjoint $k$-cliques models are equivalent for \emph{all} values of $k$. This is because the community structure plays no role in the i.i.d. setting. Therefore, Theorem~\ref{thm:k-clique_LB} holds for any value of $k$ when $q=0$, and can  thus be maximized over $k$ to obtain the best-possible bound. The maximum occurs at $k=1$ (i.e., when every vertex is its own community), which recovers the i.i.d. lower bound of $np\log(1/p)$ as desired.

\subsection{Algorithm Analysis} 

\subsubsection{Binary Splitting} 
As a direct consequence of Theorem~\ref{thm:sbm_binary_split} and Corollary~\ref{cor:sbm_binary_split}, we obtain the following non-asymptotic and asymptotic upper bounds on the expected number of tests used by binary splitting under the disjoint $k$-cliques model.

\begin{corollary}

Under the disjoint $k$-cliques model, the binary splitting algorithm identifies all infected individuals using $T$ tests, where 
			\begin{align*} 
			\EE[T] &\leq  m\cdot k \cdot \Big(\log_2 m + \log_2 k + 2\Big) \cdot \Big(1- (1-p)(1-pq)^{k-1}\Big) + 1 \\
			&\preceq \lp m \cdot k^2 \cdot p  \cdot \lp \log_2 m + \log_2 k \rp  \cdot (1/k + q) \rp. 
			\end{align*} 
\end{corollary}

\subsubsection{Community-Aware Algorithm}

The following bounds are obtained as direct consequences of Theorem~\ref{thm:sbm_graphaware} and Corollary~\ref{cor:sbm_graphaware}. 

\begin{corollary}
Under the disjoint $k$-cliques model, the community-aware algorithm identifies all infected individuals using $T$ tests, where 
			\begin{align*} 
			\EE[T] &\leq m \cdot \Big(\log_2 m + 3 \Big) \cdot \Big(1-(1-p)^k\Big) + n \cdot \Big(\log_2 k + 2\Big) \cdot \Big(1- (1-p)(1-pq)^{k-1} \Big) + 1 \\
			&\preceq m\cdot k \cdot p \cdot \log m  + m \cdot k^2 \cdot p \cdot \Big(\frac{1}{k} + q\Big) \cdot  \log k. 
			\end{align*} 
\end{corollary}

\subsection{Discussion}\label{subsec:k-cliques_discussion}

We summarize the expected number of tests of binary splitting and the community-aware algorithm, as well as the information-theoretic lower bound, in Table~\ref{tbl:k-clique}.

\begin{table}[hbt!]
\centering 
\bgroup
\def\arraystretch{2}
\begin{tabular}{|c|c|}
\hline
    Binary splitting & $m \cdot k^2 \cdot p\cdot \lp \frac{1}{k} + q\rp \cdot \log m + m\cdot k^2 \cdot p\cdot \lp \frac{1}{k} + q\rp \cdot \log k$ \\
\hline
    Community-aware & $ m\cdot k \cdot p \cdot \log m  + m \cdot k^2 \cdot p \cdot \Big(\frac{1}{k} + q\Big) \cdot  \log k  $ \\
\hline
    Lower bound & $ m \cdot k \cdot p \cdot \log\Big(\frac{1}{kp}\Big) + m\cdot k^2\cdot p\cdot q \cdot \lp\log k + \log\log\lp \frac{1}{k p} \rp\rp+1 $ \\
\hline
\end{tabular} \egroup 
\caption{Upper and lower bounds on the expected number of tests in the disjoint $k$-cliques model.}
\label{tbl:k-clique}
\end{table}

If we compare the bounds for binary splitting and the community-aware algorithm term-by-term, we observe that the binary splitting bound has an extra additive factor of $mk^2pq\log m$. Thus, as with the general SBIM, the community-aware algorithm is never worse (order-wise) than binary splitting. 

Next, we discuss different parameter regimes where 1) the lower bound holds, 2) the community-aware algorithm is order-optimal (i.e., the lower bound is tight), and 3) the community-aware algorithm's average complexity is strictly better than binary splitting's.
As stated in Theorem~\ref{thm:k-clique_LB}, the lower bound holds when $kp \preceq 1$ and $q\preceq \frac{1}{\sqrt{k\log\lp \frac{1}{kp} \rp}}$. The next corollary specifies the regime where the community-aware algorithm is order-optimal:
\begin{corollary}\label{cor:order_optimality}
Under the disjoint $k$-cliques model, the community-aware algorithm is order-optimal under the following conditions:
\begin{enumerate}
    \item $kp \preceq m^{-\beta}$ for some fixed $\beta \in \lp 0,1 \rp$,
    \item $\frac{1}{k} \preceq q \preceq \frac{1}{\sqrt{k \log\lp \frac{1}{kp} \rp}}$.
\end{enumerate}
\end{corollary}
\begin{proof}
Plugging $\log\lp\frac{1}{kp}\rp \succeq \beta \log m$ %
into the lower bound and using the fact that $k \succeq \log \lp \frac{1}{kp} \rp$ from the second condition (which implies $\log k \succeq \log\log m$) yields 
\begin{align*}
    \EE[T] & \succeq m \cdot k \cdot p \cdot \log m + m\cdot k^2\cdot p\cdot q \cdot \lp\log k + \log\log m\rp+1\\
    & \succeq m  \cdot k \cdot p \cdot \log m + m\cdot k^2\cdot p\cdot q \cdot \log k,
\end{align*}
and applying $q \succeq 1/k$ to the bound for the community-aware algorithm yields 
$$ \EE\lb T \rb \preceq  m\cdot k \cdot p \cdot \log m + m \cdot k^2\cdot p \cdot q \cdot \log k. $$
\end{proof}
Finally, using Corollary~\ref{cor:sbm_improvement} from our discussion on the general SBIM, we specify the regime where the community-aware algorithm outperforms binary splitting:
\begin{corollary}\label{cor:k-clique_improvement} 
    Under the disjoint $k$-cliques model, if $\log m \succ \log k$ and $kq \succ 1$, then the community-aware algorithm's average complexity is asymptotically strictly better than binary splitting's by a factor of $\min\Big\{kq, \, \frac{\log m}{\log k}\Big\}$.
\end{corollary}

\begin{remark}
Recall from Corollary~\ref{cor:sbm_improvement} that when $\log m \succ \log k$, $kq_1 \succ 1$, and $mkq_2 \preceq 1$ under the general SBIM, the improvement factor is also $\min\Big\{kq_1, \, \frac{\log m}{\log k}\Big\}$, matching Corollary~\ref{cor:k-clique_improvement} above. Intuitively, this is because the SBIM resembles the disjoint $k$-cliques model when $q_2$ is very small. 
\end{remark}

In Table~\ref{tbl:k-clique_regimes}, we summarize the different parameter regimes discussed so far.

\begin{table}[hbt!]
\centering 
\bgroup
\def\arraystretch{2}
\begin{tabular}{|c|c|}
\hline
    Lower bound's conditions & $kp \preceq 1$ and $q\preceq \frac{1}{\sqrt{k\log\lp \frac{1}{kp} \rp}}$ \\
\hline
    Tightness conditions &  $kp \preceq m^{-\beta}$ and
    $1 \preceq kq \preceq \sqrt{k/\log\lp \frac{1}{kp} \rp}$\\
\hline
    Improvement conditions & $\log m \succ \log k$ and $kq \succ 1$ \\ %
\hline
\end{tabular} \egroup 
\caption{Parameter regimes of interest for the disjoint $k$-cliques model.}
\label{tbl:k-clique_regimes}
\end{table}

As with the general SBIM, we find that the community-aware algorithm potentially improves
upon binary splitting when (i) there are several moderately sized communities in the network, and (ii) the transmission rate within each clique is significant. Additionally, the community-aware algorithm is order-optimal when the seeds are sparse. %

    \section{Noisy Setting}\label{sec:noisy}

In this section, we develop noise-resilient analogues to our previously discussed algorithms, which we call \textit{noisy binary splitting} (NBS) and the \textit{noisy community-aware} (NCA) algorithm.  %
We obtain bounds on the algorithms' average complexity under a high-probability recovery criterion and find that the NCA algorithm offers the same improvement in testing efficiency compared to NBS as observed between the corresponding algorithms in the noiseless case. %
We also provide the proof of our lower bound (Theorem~\ref{thm:noisy_LB}) and discuss how this result implies the order-optimality of the NCA algorithm in the same parameter regimes as the noiseless community-aware algorithm.
Thus, broadly speaking, the presence of testing noise does not affect the relative gains of using a community-oriented approach. 

\subsection{Noisy Binary Splitting}
In \cite{teo2022noisy}, an algorithm called \textit{modified noisy binary search} (MNBS)\footnote{This is essentially the noisy binary search algorithm from \cite{ben2008bayesian} adapted to the group testing framework.} was used as a sub-routine of an adaptive procedure for the symmetric noise model \eqref{eqn:noisy_model}.
The MNBS algorithm is said to succeed if it identifies an infected member of the population when one exists, and otherwise outputs a special symbol $\phi$ when no infected members remain in the population. 
We will employ the MNBS algorithm in a black-box manner and utilize the following performance guarantee \cite[Lemma 2]{teo2022noisy}: 

\begin{lemma}[MNBS Guarantee]\label{lem:mnbs}
Under the symmetric noise model \eqref{eqn:noisy_model}, given any $\delta \in (0,1)$, the MNBS algorithm succeeds with probability at least $1-\delta$ while satisfying 
\[\EE[\text{\# tests}] \leq \frac{\log n}{I(\rho)} + O\Big(\log \frac{1}{\delta}\Big) + O\Big(\log \log n\Big),\]
where $I(\rho) = 1-\msf{h}_\msf{b}(\rho) = 1 -  \rho \log_2 \frac{1}{\rho} - (1-\rho) \log\frac{1}{1-\rho}$ is the capacity of the binary symmetric channel with crossover probability $\rho \in (0, \frac{1}{2})$.%
\end{lemma}
\begin{proof}
See \cite[Appendix A]{teo2022noisy}.
\end{proof}

We obtain a very simple  \textit{noisy binary splitting} (NBS) algorithm for the symmetric model via repeated applications of MNBS. Note that NBS does not take into account the community structure of the population inherent to $\msf{SBIM}(n,k,p,q_1,q_2)$. 

\vspace{1em}  
    \noindent
	\fbox{ \parbox{0.97\textwidth}{
	\textbf{Noisy Binary Splitting (NBS) Algorithm} 
\begin{enumerate}
\setcounter{enumi}{-1}
    \item Initialize $\hat{X} = \emptyset$; $\mathcal{P} = [n]$. 
    \item Run MNBS with error parameter $\delta \in (0,1)$ on $\mathcal{P}$. If the result is $\phi$, terminate and return $\hat{X}$. Otherwise, add the result to $\hat{X}$ and remove it from $\mathcal{P}$. Repeat Step 1.
\end{enumerate}
	}}
\vspace{1em} 

\begin{remark}
If binary search algorithms are developed for other noise models (e.g., dilution noise or Z-channel noise), then NBS (and consequently the NCA algorithm discussed later) can easily be adapted to these models by replacing the MNBS algorithm with the channel-specific algorithm. Likewise, to analyze the error probability and average number of tests, we need only replace the performance guarantees of the MNBS algorithm with those of the channel-specific algorithm. 
\end{remark}

Recall that our objective is to ensure a vanishing probability of error, defined as $P_e \eqDef \Pr(\hat{X} \neq X)$. %
We obtain the following bound on the average number of tests used by NBS under this recovery guarantee.
The proof closely follows that of \cite[Theorem~1]{teo2022noisy}, the main difference being the random nature of infections in our setting, in contrast to their assumption that the number of infections is fixed.

\begin{theorem}[NBS Bound]\label{thm:nbs}
Suppose the infections in the population are distributed according to the $\msf{SBIM}(n,k,p,q_1,q_2)$, and let
\[\mu \eqDef n \cdot \Big(1-(1-p)\cdot (1-p\cdot q_1)^{k-1}\cdot (1-p\cdot q_2)^{n-k}\Big)\]
denote the expected number of infected individuals. Under the symmetric noise model \eqref{eqn:noisy_model}, for any $\delta \in (0,1)$ such that $\delta \prec \frac{1}{\mu}$, the NBS algorithm achieves 
$P_e \leq \delta\mu \to 0$ and uses an expected number of tests satisfying 
\begin{equation}
    \EE[T] \preceq \frac{\mu \log n}{I(\rho)} + \mu \log \frac{1}{\delta}. %
\end{equation}
\end{theorem}
\begin{proof}
We first analyze the probability of error. Note that, as long as the MNBS algorithm always succeeds in Step 1, the NBS algorithm will produce the correct output (i.e., achieve $\hat{X} = X$) and do so using $D+1$ calls to the MNBS algorithm, where $D$ is the total number of infected individuals in the population. (The $+1$ is due to the final call to the algorithm when no infected members remain.) Let $\mathcal{E}_i$ denote the event that the $i^\text{th}$ call to the MNBS algorithm fails. Then, by conditioning on the number of infections and performing a union bound over the error events, we have 
\begin{align*}
    P_e \leq \Pr\Bigg(\bigcup_{i=1}^D \mathcal{E}_i\Bigg) &= \sum_{d=0}^n \Pr(D = d) \cdot \Pr\Bigg(\bigcup_{i=1}^D \mathcal{E}_i\, \Big\vert \, D=d\Bigg) \\
    &\leq \sum_{d=0}^n \Pr(D=d) \cdot d \delta  \\
    &= \delta \mu \, \to \, 0 \text{ (by assumption).}
\end{align*}

Next, we bound the average number of tests used by the algorithm. By Lemma~\ref{lem:mnbs}, the average number of tests performed during the first $D+1$ calls to the MNBS algorithm is 
\begin{equation}\label{eqn:nbs_success}
\Big(\mu +1\Big) \Bigg(\frac{\log n}{I(\rho)} + O\Big(\log \frac{1}{\delta} \Big) + O(\log\log n)\Bigg) \asymp \frac{\mu \log n}{I(\rho)} + \mu \log \frac{1}{\delta}.
\end{equation}
Note that the algorithm may make more than $D+1$ calls to MNBS if at least one of the calls fails. However, these additional tests do not affect the overall scaling in \eqref{eqn:nbs_success} for the following reasons. Since the tests are subject to independent noise, the calls to MNBS fail independently with probability $O(\delta)$. Hence, the number of failures encountered before the first success is distributed as a $\msf{Geometric}(\theta)$ random variable where $\theta = 1-O(\delta)$. The mean of this variable is $\frac{1-\theta}{\theta} = O(\delta)$, and it follows that $\EE[T]$ is within a multiplicative $1+O(\delta)$ factor of \eqref{eqn:nbs_success}.
\end{proof}

\begin{remark}\label{rmk:nbs}
Note that the bound in Corollary~\ref{cor:sbm_binary_split} (the average number of tests used by binary splitting in the noiseless case) can be expressed as $\mu \log n$. The bound in Theorem~\ref{thm:nbs} scales this inversely by $I(\rho)$ (the capacity of the binary symmetric channel) and includes an additive factor, $\mu \log \frac{1}{\delta}$, that depends on the desired error probability.

We further remark that Theorem~\ref{thm:nbs} applies to any underlying probabilistic infection model, not just the $\msf{SBIM}$. One just needs to replace $\mu$ in the theorem statement with the corresponding model-specific quantities.
\end{remark} 

\subsection{Noisy Community-Aware Algorithm}
Following the same template as the noiseless community-aware algorithm, our \textit{noisy community-aware} (NCA) algorithm applies NBS in two stages: first to identify communities with at least one infected member, and then to identify the infected individuals within each positive community from the first stage. 

\vspace{1em}  
    \noindent
	\fbox{ \parbox{0.97\textwidth}{
	\textbf{Noisy Community-Aware (NCA) Algorithm} 
\begin{enumerate}
    \item Mix the samples within each of the %
    communities %
    and run the NBS algorithm with error parameter $\delta' \in (0,1)$ on the mixed samples.
    \item Run the NBS algorithm with error parameter $\delta \in (0,1)$ within each community identified as positive in Step 1. 
\end{enumerate}
	}}
\vspace{1em} 

We now state and prove the performance guarantee of the NCA algorithm.

\begin{theorem}[NCA Bound]\label{thm:nca}
Suppose the infections in the population are distributed according to $\msf{SBIM}(n,k,p,q_1,q_2)$, and let
\begin{align*}
    &\mu \eqDef n \cdot \Big(1-(1-p)\cdot (1-p\cdot q_1)^{k-1}\cdot (1-p\cdot q_2)^{n-k}\Big) \\
    &\mu_M \eqDef m \cdot \Bigg(1-(1-p)^k \cdot \Big(1-p\cdot(1-(1-q_2)^k)\Big)^{n-k}\Bigg) 
\end{align*}
denote the expected number of infected individuals and expected number of infected communities, respectively.
Under the symmetric noise model \eqref{eqn:noisy_model}, for any $\delta, \delta' \in (0,1)$ such that $\delta' \prec \frac{1}{\mu_M}$ and $\delta \prec \frac{1}{\mu}$, the NCA algorithm achieves 
$P_e \leq \delta'\mu_M + \delta\mu \to 0$ and uses an expected number of tests satisfying 
\begin{equation}\label{eqn:nca}
    \EE[T] \preceq \frac{\mu_M \log m + \mu \log k}{I(\rho)} + \mu \log \frac{1}{\delta}.
\end{equation}
\end{theorem}
\begin{proof}
Provided that each step of the algorithm succeeds, the algorithm's final output will be correct. Moreover, the NBS algorithm will be executed once in Step 1 and $M$ times in Step 2, where $M$ is the number of infected communities. Let $\calE^{(1)}$ denote the event that the single call to NBS in Step 1 fails, and let $\calE^{(2)}_i$ denote the event that the $i^\text{th}$ call to NBS in Step 2 fails. Additionally, let $\mu_K$ denote the expected number of infected members in a single community conditioned on the event that the community contains at least one infected member. We have 
\begin{align}
    P_e &\leq \Pr\Bigg(\calE^{(1)} \cup \bigcup_{i=1}^M \calE^{(2)}_i\Bigg) \\
    &\leq \Pr(\calE^{(1)}) + \sum_{j=0}^m \Pr(M=j) \cdot \Pr\Bigg(\bigcup_{i=1}^M \calE^{(2)}_i \, \Big\vert \, M=j\Bigg) \label{eqn:nca_pf1} \\
    &\leq \delta' \cdot \mu_M + \sum_{j=0}^m \Pr(M=j)\cdot j\cdot \delta \cdot \mu_K \label{eqn:nca_pf2} \\
    &= \delta' \cdot \mu_M + \delta \cdot \mu_K \cdot \mu_M \\
    &= \delta'\cdot \mu_M + \delta \cdot \mu \, \to \, 0 \label{eqn:nca_pf3}
\end{align}
where \eqref{eqn:nca_pf1} applies a union bound and conditions on the number of infected communities, \eqref{eqn:nca_pf2} applies another union bound along with the NBS error guarantee from Theorem~\ref{thm:nbs}, and \eqref{eqn:nca_pf3} uses the identity $\mu = \mu_K \cdot \mu_M$. 

To obtain a bound on the average number of tests, we will apply Theorem~\ref{thm:nbs} twice and sum the results. First, we apply it with $\mu$ replaced by $\mu_M$, $n$ replaced by $m$, and $\delta$ replaced by $\delta'$ (corresponding to Step 1). Second, we apply it with $\mu$ replaced by $\mu_K$ and $n$ replaced by $k$ (corresponding to Step 2). %
The overall average number of tests performed during a single call to the NBS algorithm in Step 1 and $M$ calls to the algorithm in Step 2 is thus given by 
\begin{align*}
    \EE[T] &\preceq \frac{\mu_M \log m}{I(\rho)} + \mu_M \log \frac{1}{\delta'} + \mu_M \Bigg(\frac{\mu_K \log k}{I(\rho)} + \mu_K \log \frac{1}{\delta} \Bigg) \\
    &= \frac{\mu_M \log m}{I(\rho)} + \mu_M \log \frac{1}{\delta'} + \frac{\mu\log k}{I(\rho)} + \mu \log \frac{1}{\delta} \\
    &\preceq \frac{\mu_M \log m + \mu \log k}{I(\rho)} + \mu \log\frac{1}{\delta} 
\end{align*}
where the second line uses the identity $\mu = \mu_K \cdot \mu_M$ and the final line uses the fact that $\mu_M \log \frac{1}{\delta'} \preceq \mu \log \frac{1}{\delta}$. Again, we note that additional tests may be required if one or more of the calls to the NBS algorithm fails, but that the above scaling will remain intact. (See the argument in the proof of Theorem~\ref{thm:nbs}.)

\end{proof}

\begin{remark}
The bound in Corollary~\ref{cor:sbm_graphaware} (the community-aware algorithm in the noiseless case) can be expressed as $\mu_M \log m + \mu \log k$. Hence, as similarly discussed in Remark~\ref{rmk:nbs}, we find that Theorem~\ref{thm:nca} scales the noiseless bound inversely by $I(\rho)$ and includes an additive factor that depends on the desired error probability $\delta$.

It follows that the NCA algorithm is always at least as efficient (order-wise) as NBS (recall our discussion in Section~\ref{subsec:sbm_discussion}). Moreover, if $\delta \succeq k^{-\frac{1}{I(\rho)}}$ (which ensures that the $\mu \log \frac{1}{\delta}$ term does not dominate), then the NCA algorithm is strictly better than NBS under the same conditions we derived in 
Corollary~\ref{cor:sbm_improvement}. That is, the parameter regimes in which the community-aware algorithm improves upon binary splitting in the noiseless case are the same regimes in which the NCA algorithm improves upon NBS, provided that the desired error probability is not too stringent. 

Lastly, we emphasize the generality of Theorem~\ref{thm:nca} beyond the $\msf{SBIM}$. The result holds under any probabilistic infection model satisfying the following symmetry condition, which states that the expected number of infections in an infected community should be the same across all communities: 
\[\EE\Big[\sum_{\ell \in \calC_i} X_\ell \, | \, X_{\calC_i} = 1\Big] = \EE\Big[\sum_{\ell \in \calC_j} X_\ell \, | \, X_{\calC_j} = 1\Big], \quad \forall i,j \in [m].\]
Also, note that our result holds even when there are no ``communities'' in the population. In this case, we simply set $m = 1,\, k=n$ in \eqref{eqn:nca} and find that it reduces to the community-oblivious bound in Theorem~\ref{thm:nbs}. 
\end{remark}

\subsection{Proof of the Lower Bound (Theorem~\ref{thm:noisy_LB})} 
Before delving into the proof of Theorem~\ref{thm:noisy_LB}, we make a few remarks. Provided that $\delta \succeq k^{-1/I(\rho)}$, note that the NBS upper bound (Theorem~\ref{thm:nbs}) and the NCA upper bound (Theorem~\ref{thm:nca}) each scale their noiseless counterpart (Theorem~\ref{thm:sbm_binary_split} and Theorem~\ref{thm:sbm_graphaware}, respectively) by $I(\rho)$. Moreover, our lower bound in Theorem~\ref{thm:noisy_LB} also scales its noiseless counterpart \eqref{eqn:entropy_LB} by $I(\rho)$. Taken together, these results imply that the NCA algorithm is order-optimal under the same conditions as the noiseless community-aware algorithm (recall Corollary~\ref{cor:order_optimality}), as long as $\delta$ is not too small. Thus, the benefits of using a community-aware scheme persist in the presence of symmetric testing noise.

To prove our main result, we will leverage a slightly more general version of Fano's inequality given in the following lemma. %
\begin{lemma}\label{lem:fano}
For any estimator $\hat{X}$ such that $X \to Y \to \hat{X}$, with $P_e = \Pr(X \neq \hat{X})$, we have
\[H(X\, | \, Y) \leq 1 + P_e \cdot H(X).\]
\end{lemma}
\begin{proof}
The proof is a small modification to the proof of Fano's inequality \cite[Theorem~2.10.1]{cover2006elements}. First, we define the error random variable
\[ E = 
\begin{cases}
1 & \text{if } \hat{X} \neq X \\
0 & \text{if } \hat{X} = X. 
\end{cases}
\]
Then, we expand $H(X,E \, | \, \hat{X})$ using two applications of the chain rule for entropy: 
\begin{align*}
    H(X,E \, | \, \hat{X}) &= H(X \, | \, \hat{X}) + \underbrace{H(E\,|\,X,\hat{X})}_{=0} \\
    &= H(E\,|\,\hat{X}) + H(X\,|\,E,\hat{X}).
\end{align*}
Since conditioning reduces entropy, we have $H(E\,|\,\hat{X}) \leq H(E) = \mathsf{h}_\mathsf{b}(P_e) \leq 1.$ Therefore, we currently have
\[H(X\,|\,\hat{X}) \leq 1 + H(X\,|\,E,\hat{X})\]
and we can bound $H(X\,|\,E,\hat{X})$ as
\begin{align*}
    H(X\,|\,E,\hat{X}) &= \Pr(E=0)\cdot \underbrace{H(X\,|\,\hat{X},E=0)}_{=0} + \Pr(E=1)\cdot H(X\,|\,\hat{X}, E=1) \\
    &\leq P_e \cdot H(X)
\end{align*}
where the inequality follows from the fact that conditioning reduces entropy. Finally, the data processing inequality yields
\[H(X\,|\,Y) \leq H(X\,|\,\hat{X}) \leq 1 + P_e \cdot H(X)\]
as desired. 
\end{proof}

Now we prove our main result, which we restate below for convenience.

\newtheorem*{T1}{Theorem~\ref{thm:noisy_LB}}

\begin{T1}
    Assume $H(X_1,X_2,\ldots, X_n) \to \infty$ as $n \to \infty$. Under the symmetric noise model \eqref{eqn:noisy_model}, any adaptive algorithm achieving $P_e \to 0$ must use an average number of tests lower bounded as 
    \begin{equation}
        \EE[T] \geq \frac{H(X_1,\ldots, X_n)}{I(\rho)}.
    \end{equation}
\end{T1}

\begin{proof}
For $i=1,2,\dots, M$, where $M$ is the maximum number of tests allowed by the algorithm, let $\calS_i \subseteq [n]$ denote the subset of individuals included in the $i^\text{th}$ test. Further, let 
\[Z_i = \bigvee_{j \in \calS_i} X_j\] 
denote the \textit{noiseless} outcome of the $i^\text{th}$ test, and let  
\[Y_i = Z_i \oplus \xi, \quad \xi \sim \msf{Bernoulli}(\rho), \, \rho \in \Big(0,\frac{1}{2}\Big)\]
denote the corresponding noisy outcome. An adaptive algorithm applies a sequence of tests $\calS_1,\calS_2, \dots$ and observes their outcomes $Y_1, Y_2,\dots$, where $\calS_i$ can be chosen as a function of the previous test outcomes $Y_1,\dots, Y_{i-1}$. The algorithm terminates at a random time $T$ which is determined by the test outcomes up to the $T^\text{th}$ test, i.e., $Y_1,\dots, Y_T$. For $i > T$, we will write $\calS_i = \emptyset$, $Z_i = \emptyset$, and $Y_i = \emptyset$ to indicate that the algorithm has been terminated and no further tests will be performed. %

We proceed by expanding the mutual information between the infection statuses $X_1,\ldots, X_n$ and the noisy test outcomes $Y_1,\ldots, Y_M$ as follows: 
\begin{align}
    I(X_1,\ldots, X_n; Y_1,\ldots Y_M) &= H(Y_1,\ldots, Y_M) - H(Y_1, \ldots, Y_M \, | \, X_1,\ldots, X_n) \nonumber \\
    &= \sum_{i=1}^M H(Y_i \, | \, Y_1,\ldots, Y_{i-1}) - \sum_{i=1}^M H(Y_i \, | \, X_1,\ldots, X_n, Y_1,\ldots, Y_{i-1}).  \label{eqn:mutual_info}
\end{align}
For now, let us assume that the algorithm is deterministic and that the first test, $\calS_1$, is fixed. Note that $\calS_i$ can be recursively deduced from $Y_1,\ldots, Y_{i-1}$, and hence $\calS_i$ is a function of $Y_1,\ldots, Y_{i-1}$. Combining this with the fact that conditioning reduces entropy, we obtain  
\begin{equation}\label{eqn:cond_entropy0}
    H(Y_i \, | \, Y_1,\ldots, Y_{i-1}) = H(Y_i \, | \, \calS_i, Y_1,\ldots, Y_{i-1}) \leq H(Y_i \, | \, \calS_i). 
\end{equation}

Next, we have 
\begin{align}
    H(Y_i \, | \, X_1,\ldots, X_n, Y_1,\ldots, Y_{i-1}) &= H(Y_i \, | \, X_1,\ldots, X_n, Y_1,\ldots, Y_{i-1}, Z_i) \label{eqn:cond_entropy1}\\
    &= H(Y_i \, | \, Z_i) \label{eqn:cond_entropy2} 
\end{align}
where \eqref{eqn:cond_entropy1} uses the fact that %
$Z_i$ is a function of $\{X_1,\ldots, X_n, \calS_i\}$ and hence is a function of $\{X_1,\ldots, X_n, Y_1\ldots, Y_{i-1}\}$; and \eqref{eqn:cond_entropy2} uses the fact that $Y_i$ is conditionally independent of $X_1,\ldots,X_n$ and $Y_1,\ldots, Y_{i-1}$ given $Z_i$. 

Plugging \eqref{eqn:cond_entropy0} and \eqref{eqn:cond_entropy1} into \eqref{eqn:mutual_info} yields
\begin{align}
    I(X_1,\ldots, X_n; Y_1,\ldots, Y_M) &\leq \sum_{i=1}^M H(Y_i\,|\,\calS_i) - \sum_{i=1}^M H(Y_i \,|\, Z_i) \nonumber \\
    &= \sum_{i=1}^M \Pr(\calS_i \neq \emptyset) \cdot H(Y_i\,|\,\calS_i \neq \emptyset) - \sum_{i=1}^M \Pr(Z_i\neq \emptyset) \cdot \msf{h}_\msf{b}(\rho) \nonumber \\
    &\leq \sum_{i=1}^M \Pr(\calS_i \neq \emptyset) \cdot (1-\msf{h}_\msf{b}(\rho)) \label{eqn:mutual_info1} \\
    &= \sum_{i=1}^M \Pr(T \geq i) \cdot (1-\msf{h}_\msf{b}(\rho)) \nonumber \\
    &= \EE[T] \cdot (1-\msf{h}_\msf{b}(\rho)) = \EE[T] \cdot I(\rho) \label{eqn:mutual_info2}
\end{align}
where \eqref{eqn:mutual_info1} uses the fact that $H(Y_i\,|\,\calS_i\neq \emptyset) \leq 1$ since $Y_i$ is a binary random variable conditioned on $\calS_i \neq \emptyset$, and \eqref{eqn:mutual_info2} follows from the tail-sum formula for expectation. 

Now, combining \eqref{eqn:mutual_info2} and our modified Fano's inequality (Lemma~\ref{lem:fano}) with the fact that $H(X_1,\ldots, X_n) = I(X_1,\ldots, X_n; Y_1,\ldots, Y_M) + H(X_1,\ldots, X_n\,|\,Y_1,\ldots, Y_M)$ gives us
\[H(X_1,\ldots, X_n) \leq \EE[T]\cdot I(\rho) + P_e \cdot H(X_1,\ldots, X_n) + 1\]
which can be rearranged to 
\[P_e \geq 1 - \frac{\EE[T]\cdot I(\rho) + 1}{H(X_1,\ldots, X_n)}.\]
Finally, let $\varepsilon > 0$, and observe that if $\EE[T] \leq \frac{H(X_1,\ldots, X_n)}{I(\rho)}\cdot (1-\varepsilon)$, then $P_e \geq \varepsilon - \frac{1}{H(X_1,\ldots, X_n)} \to \varepsilon$ by the assumption that $H(X_1,\ldots, H_n) \to \infty$ as $n \to \infty$. 

We can easily extend our analysis to allow for randomized tests by assuming $\calS_i$ depends on some external randomness $\theta_i$, independent of the noise and the infection statuses. Thus, $\calS_i$ is a function of $\{Y_1,\ldots, Y_{i-1}, \theta_1,\ldots, \theta_{i-1}\}$. Then, using the shorthand notation $X = (X_1,\ldots, X_n)$, $Y = (Y_1,\ldots, Y_M)$, $\theta = (\theta_1,\ldots, \theta_M)$, we have 
\begin{align*} 
H(X) = H(X\,|\,\theta) &= I(X; Y \,|\, \theta) + H(X\,|\,Y,\theta) \\
&= H(Y \, | \, \theta) - H(Y \, | \, X, \theta) + H(X\,|\,Y,\theta). 
\end{align*} 
Note that 
\[H(X\,|\,Y, \theta) \leq H(X\,|\,Y) \leq 1 + P_e \cdot H(X).\]

Additionally, using the fact that $\calS_i$ is a function of $\{Y_1,\ldots, Y_{i-1}, \theta_1,\ldots, \theta_i\}$, it follows that 
\begin{align*}
    H(Y \, | \, \theta) &= \sum_{i=1}^M H(Y_i \, | \, Y_1,\ldots, Y_{i-1}, \theta) \\
    &= \sum_{i=1}^M H(Y_i \, | \, Y_1,\ldots, Y_{i-1}, \theta, \calS_i) \\
    &\leq \sum_{i=1}^M H(Y_i \,|\, \calS_i). 
\end{align*}

Finally, 
\begin{align*}
    H(Y \,|\,X, \theta) &= \sum_{i=1}^M H(Y_i \,|\, X, Y_1,\ldots Y_{i-1}, \theta) \\
    &= \sum_{i=1}^M H(Y_i \,|\, X, Y_1,\ldots, Y_{i-1}, \theta, Z_i) \\
    &= \sum_{i=1}^M H(Y_i\,|\,Z_i)
\end{align*}
and the remainder of the proof is exactly the same as before. 
\end{proof}
    
    \section{Numerical Experiments}\label{sec:simulations} 

\begin{figure}[t]
    \centering
    \begin{subfigure}[t]{0.45\textwidth}
    \includegraphics[width=\textwidth]{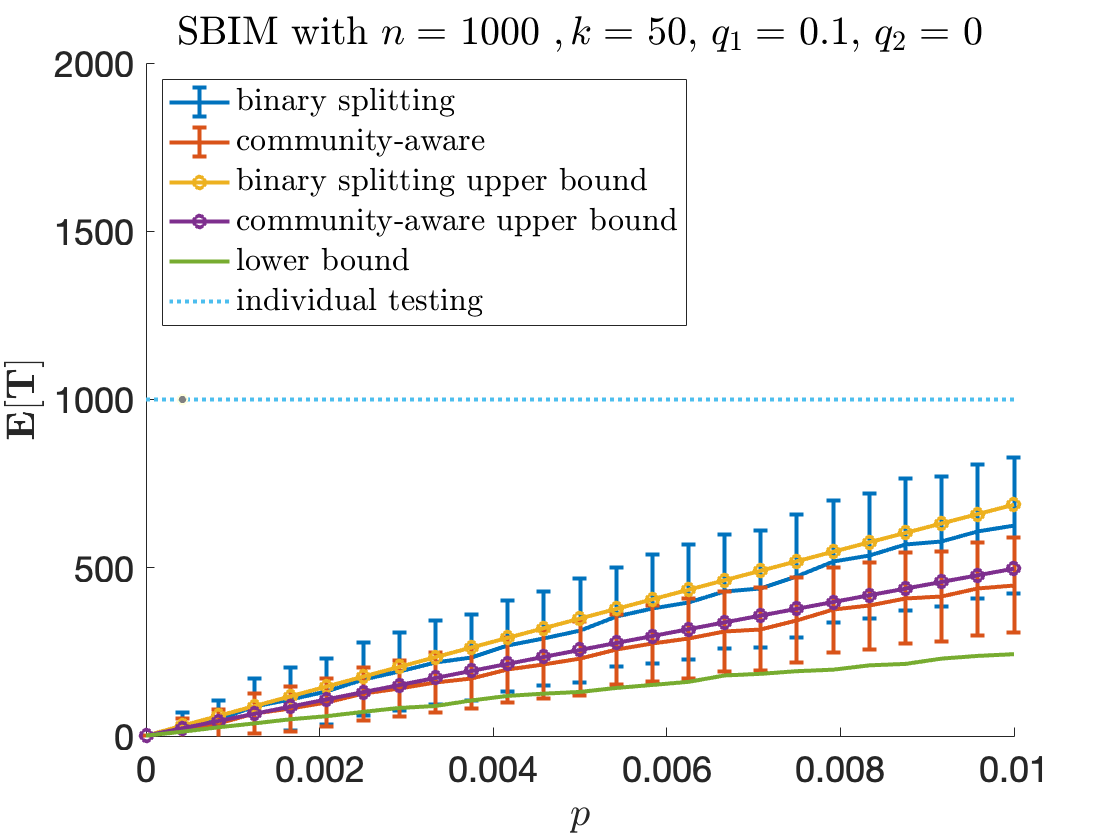}
    \caption{}
    \end{subfigure} 
    \begin{subfigure}[t]{0.45\textwidth}
    \includegraphics[width=\textwidth]{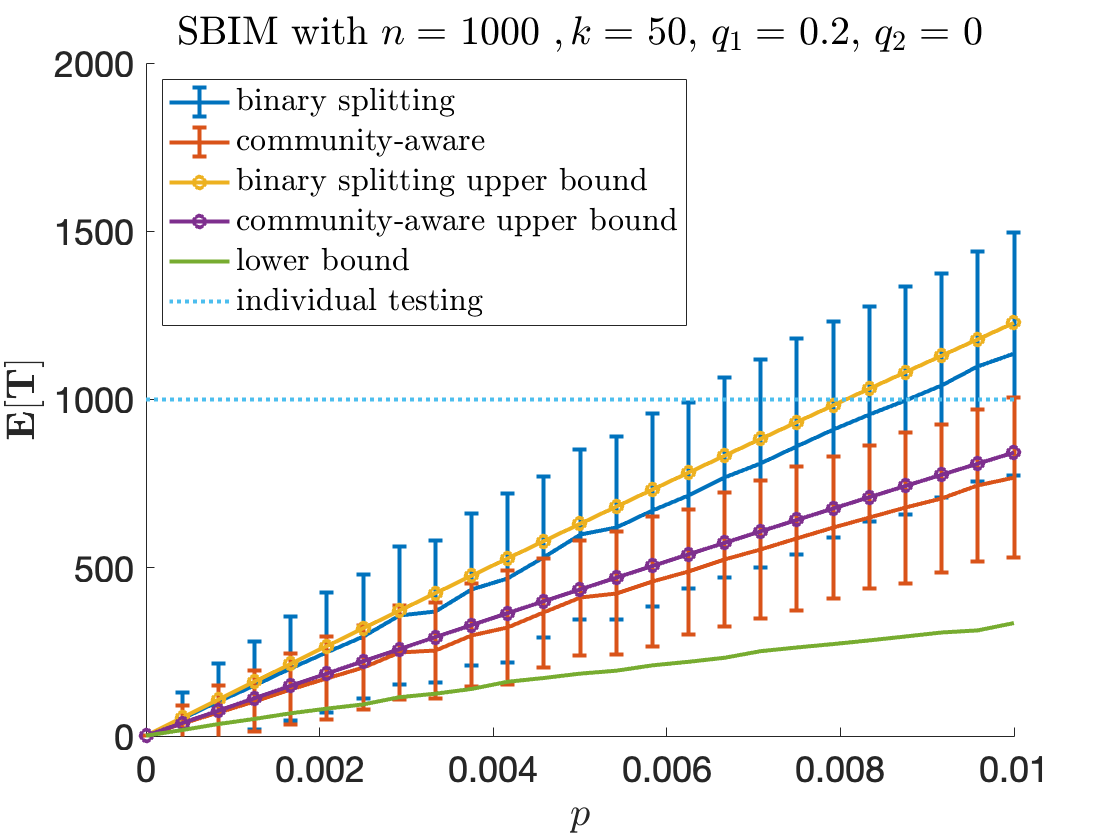}
    \caption{}
    \label{fig:sims1_b} 
    \end{subfigure} 
    \begin{subfigure}[t]{0.45\textwidth}
    \includegraphics[width=\textwidth]{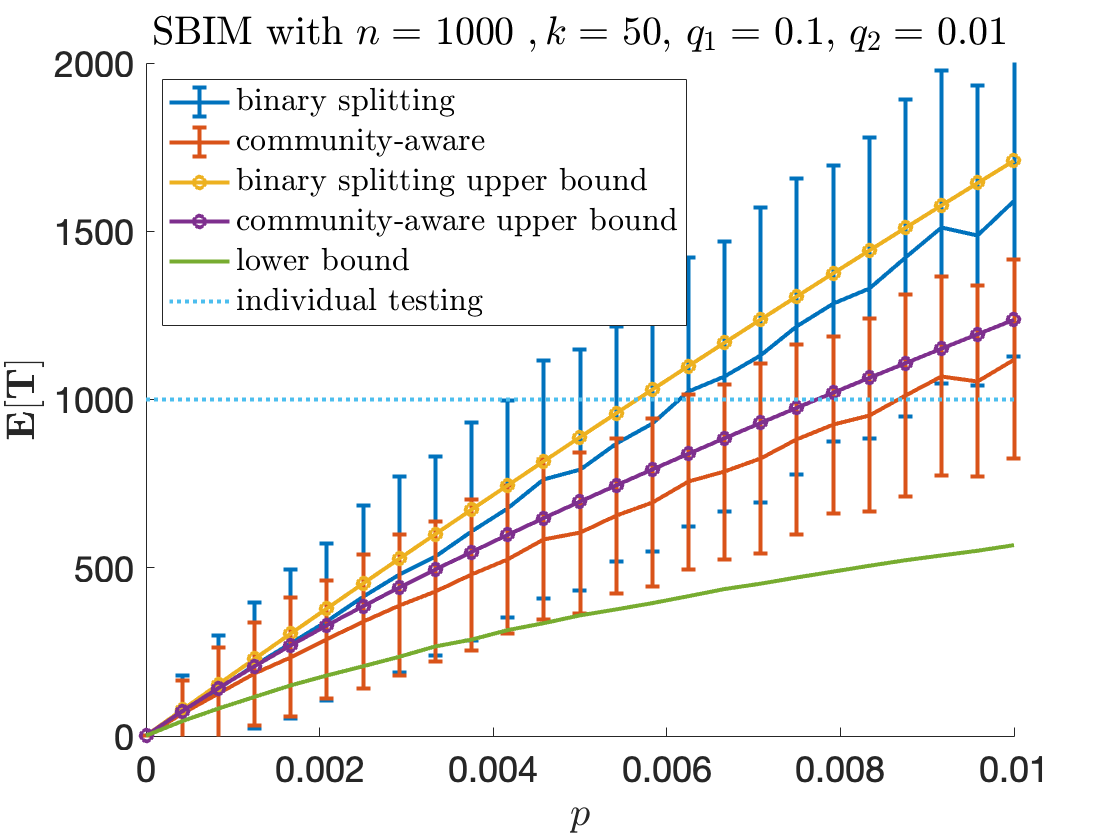}
    \caption{}
    \label{fig:sims1_c} 
    \end{subfigure}
    \begin{subfigure}[t]{0.45\textwidth}
    \includegraphics[width=\textwidth]{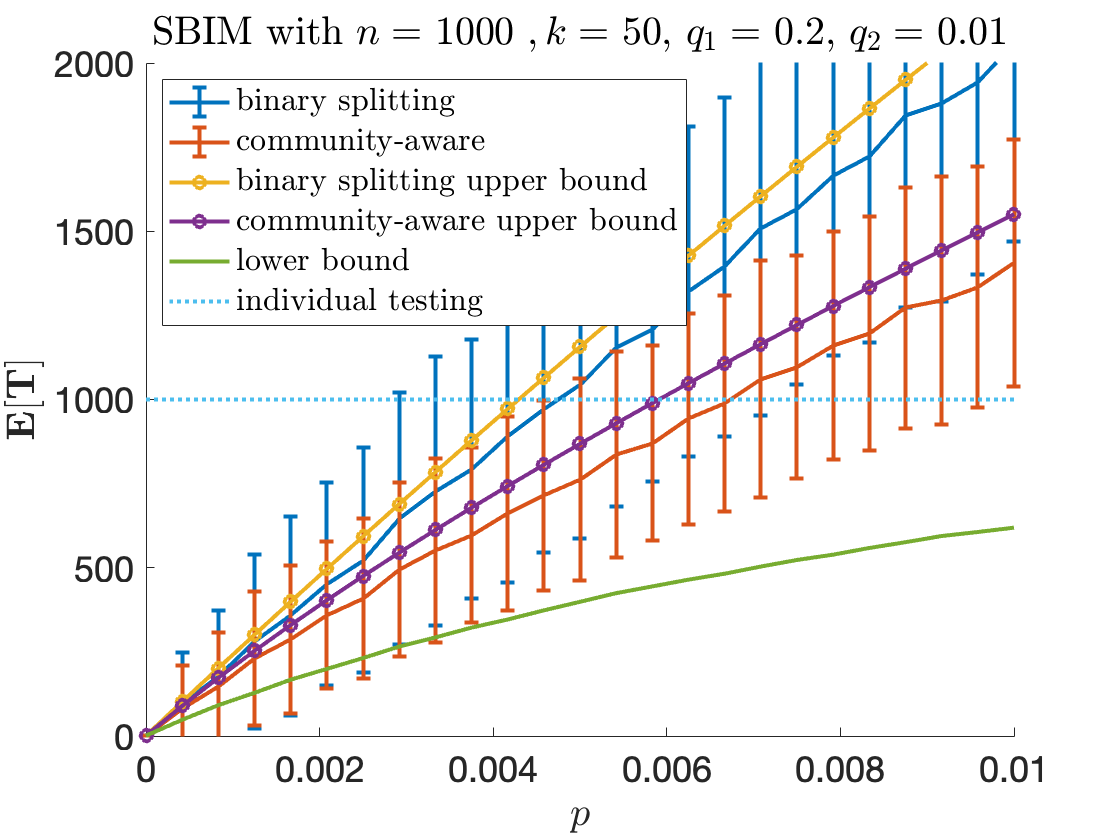}
    \caption{}
    \label{fig:sims1_d} 
    \end{subfigure}
    \caption{Performance comparison between binary splitting and the community-aware algorithm under the $\msf{SBIM}$ with $n = 1000, \, k = 50$, and different values of $p,q_1, q_2$. Theoretical upper and lower bounds are also shown. }
    \label{fig:sims1}
\end{figure}

\begin{figure}[t]
    \centering
    \begin{subfigure}[t]{0.45\textwidth}
    \includegraphics[width=\textwidth]{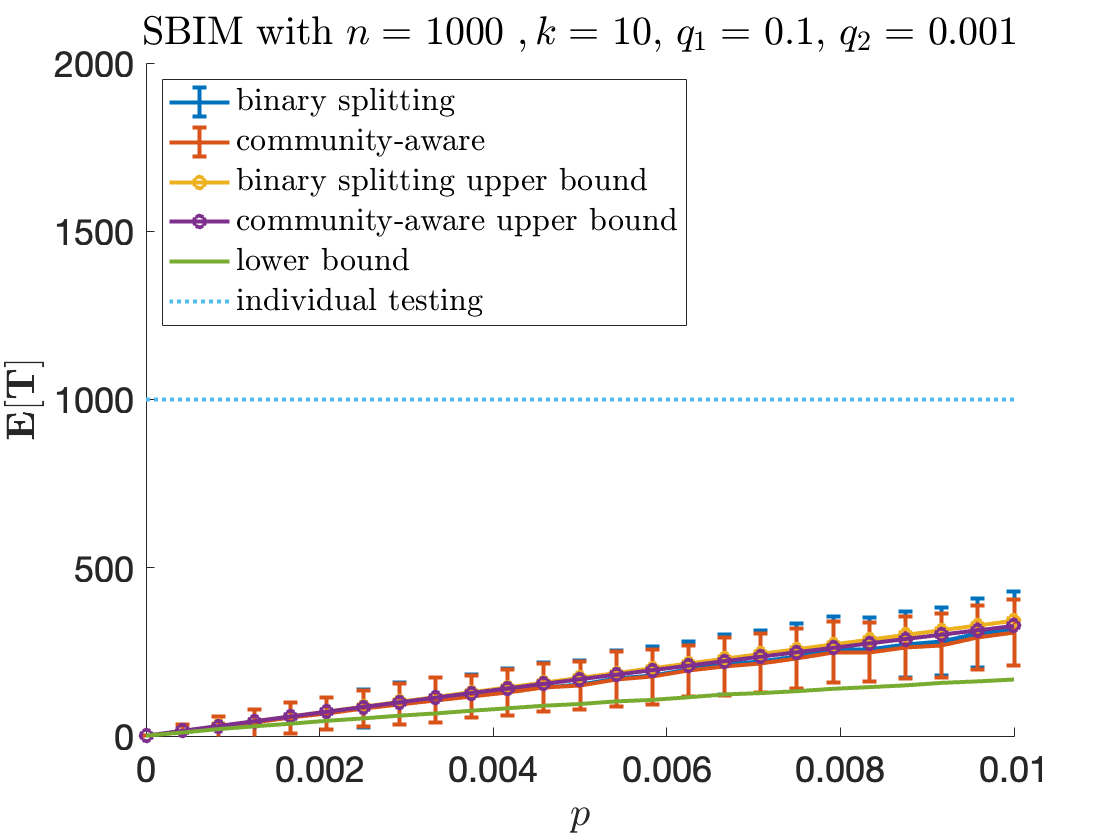}
    \caption{}
    \end{subfigure} 
    \begin{subfigure}[t]{0.45\textwidth}
    \includegraphics[width=\textwidth]{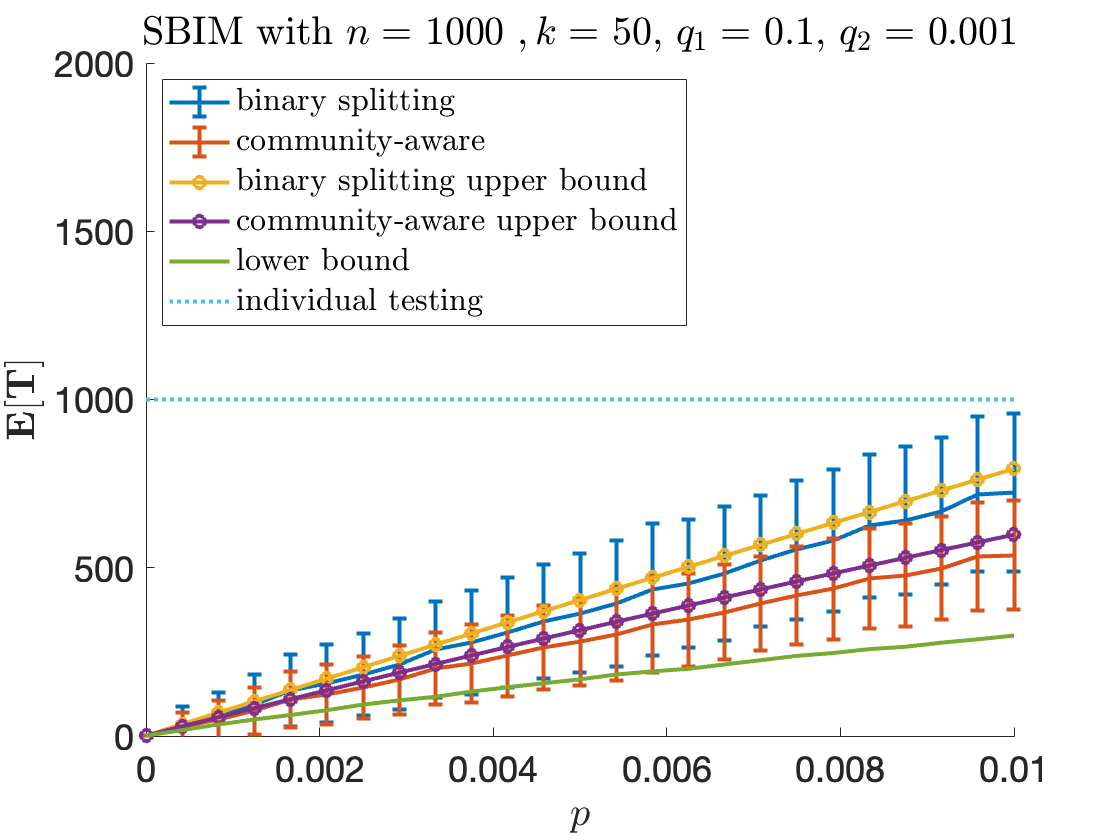}
    \caption{}
    \end{subfigure} 
    \begin{subfigure}[t]{0.45\textwidth}
    \includegraphics[width=\textwidth]{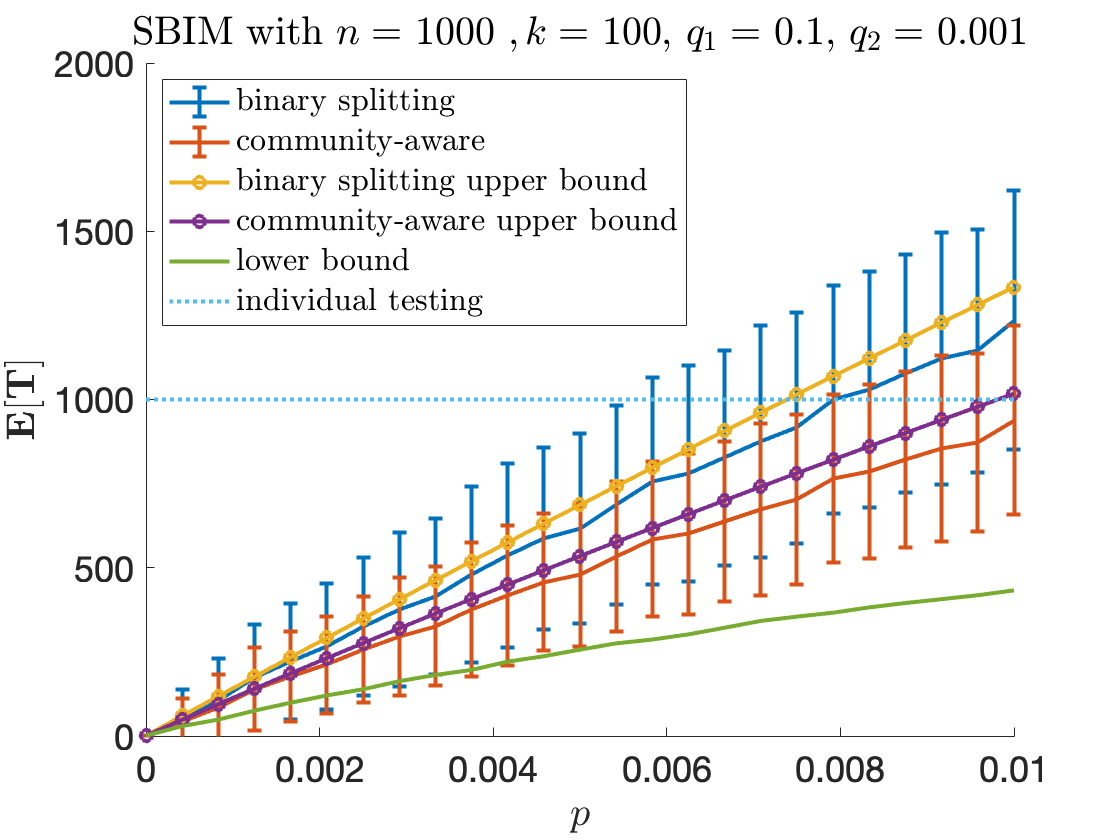}
    \caption{}
    \end{subfigure}
    \begin{subfigure}[t]{0.45\textwidth}
    \includegraphics[width=\textwidth]{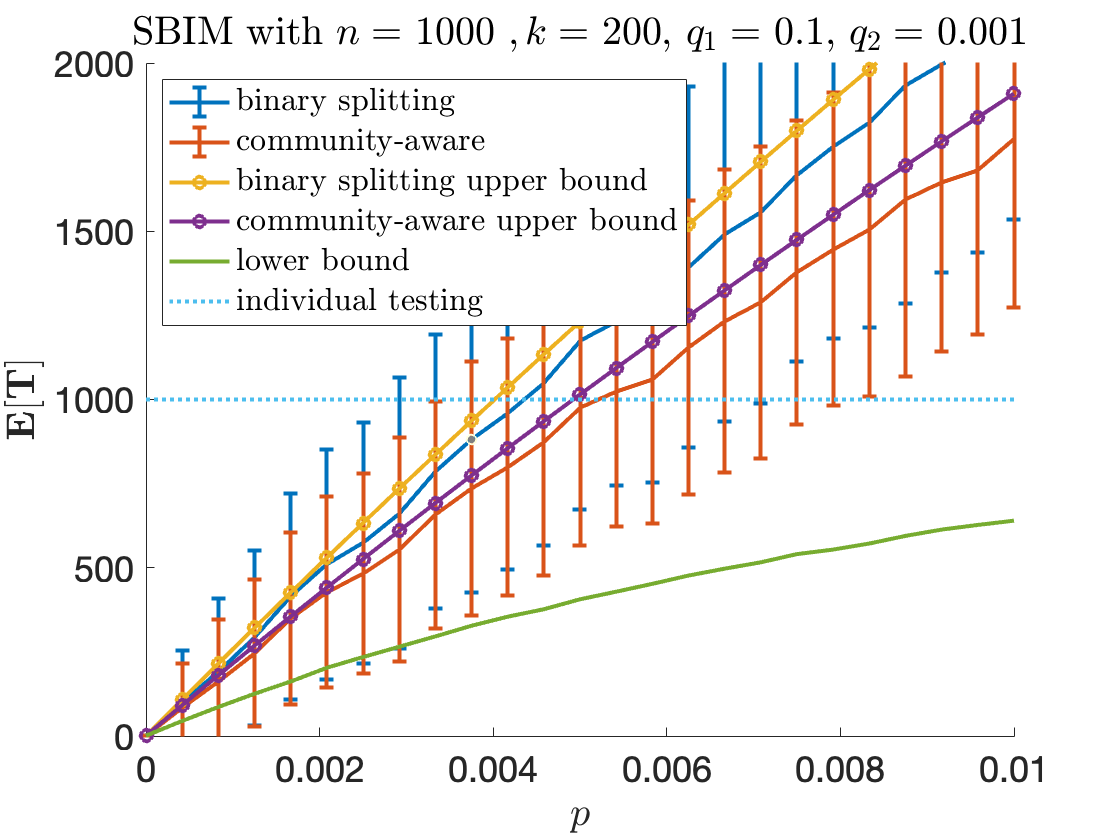}
    \caption{}
    \end{subfigure} 
    \caption{Performance comparison between binary splitting and the community-aware algorithm under the $\msf{SBIM}$ with $n = 1000, \, q_1 = 0.1, \, q_2 = 0.001$, and different values of $p, k$. Theoretical upper and lower bounds are also shown.}
    \label{fig:sims2}
\end{figure}

We implemented the binary splitting and community-aware algorithms and evaluated their performance over random instances of the SBIM. The population size was set to $n = 1{,}000$, and $p$ was varied over the interval $[0,\, 0.01]$, while $q_1,q_2$ were fixed at different values. We ran $500$ trials for each value of $p$, where a trial consists of generating an instance from $\mathsf{SBIM}(n,k,p,q_1,q_2)$ and observing the number of tests used by binary splitting and the community-aware algorithm to identify the infected individuals. For the lower bound, we computed the non-asymptotic expression given in \eqref{eqn:full_LB}. To do so, one can expand \eqref{eqn:full_LB} as follows:  
\begin{align*}
    n\cdot I(X_1;S_1) + H(X\,|\,S) &= n\cdot\Bigg(H(X_1) - H(X_1\,|\,S_1)\Bigg) + H(X\,|\,S) \\
    &= n\cdot\Bigg(H(X_1) - \PP(S_1=0)\cdot H(X_1\,|\,S_1=0)\Bigg) + H(X\,|\,S) \\
    &= n\cdot\Bigg(\mathsf{h}_\mathsf{b}\Big(\PP(X_1=1)\Big) - (1-p)\cdot \mathsf{h}_\mathsf{b}\Big(\PP(X_1=1\,|\,S_1=0)\Big)\Bigg) + H(X\,|\,S) 
\end{align*}
where $X=(X_1,\ldots,X_n), \, S=(S_1,\ldots,S_n)$. The terms $\PP(X_1=1)$ and $\PP(X_1=1\,|\,S_1=0)$ are straightforward to compute (see Lemma~\ref{lem:sbm_marginal} and its proof in Appendix~\ref{app:sbm_marginal}). The term $H(X\,|\,S)$ is lower-bounded by \eqref{eqn:cond_entropy_LB}. To estimate \eqref{eqn:cond_entropy_LB}, we took an average over many independent samples of $Z \sim \msf{Binom}(k,p)$ and $Z' \sim \msf{Binom}(n-k,\, p)$.

Figure~\ref{fig:sims1} shows some representative plots of the estimated $\EE[T]$ as a function of $p$, with $k = 50$ and different values of $q_1, \, q_2$. The error bars show $\pm$ one standard deviation of the values of $T$ obtained for a particular value of $p$. For comparison, we also plot the theoretical upper bounds from Theorem~\ref{thm:sbm_binary_split} and Theorem~\ref{thm:sbm_graphaware}, and we find that these bounds closely match the empirical results. Additionally, the community-aware algorithm consistently outperforms binary splitting. For example, in Figure~\ref{fig:sims1_c}, at $p \approx 0.009$, binary splitting has already exceeded the individual testing threshold with an average of  $1{,}044$ tests, whereas the community-aware algorithm uses an average of $708$ tests; this represents a $32\%$ reduction in testing.  The community-aware algorithm's performance also seems to exhibit lower variance than binary splitting. In Figure~\ref{fig:sims2}, we fix $q_1 = 0.1, \, q_2 = 0.001$, and vary the community size $k \in \{10, \, 50, \, 100, \, 200\}$. The community-aware algorithm appears to perform most favorably (relative to binary splitting) for more intermediate values of $k$, i.e., when there are several moderately sized communities in the network. These findings are consistent with our earlier theoretical results.

The estimated lower bound is fairly close to the community-aware algorithm's bound in the regime where the seeds are very sparse (small $p$) and the network exhibits strong community structure (intermediate $k, q_1$; small $q_2$). This corroborates our analysis from Section~\ref{subsec:k-cliques_discussion}.  In other regimes---such as when $p$, $k$, or $q_2$ are larger---there is a sizable gap between the community-aware bound and the lower bound. However, this gap seems to be at most a constant factor in many cases, suggesting that the order-optimality of the community-aware algorithm still holds in broader regimes. Nevertheless, these results suggest the potential to further improve the non-asymptotic upper or lower bounds.

    \section{Conclusion}\label{sec:conclusion}
In this paper, we investigated the group testing problem over networks with community structure. Motivated by infectious diseases such as COVID-19, we proposed a network-based infection model which generalizes the traditional i.i.d. group testing model to settings in which interactions between individuals dictate the disease spread. Our proposed adaptive algorithm, which exploits the known community structure of the underlying graph, provably outperforms the community-oblivious binary splitting algorithm and is  order-optimal in certain parameter regimes, as implied by our novel lower bounds based on the system entropy. Even in the presence of symmetric noise, our community-oriented approach offers the same gains in testing efficiency as it does in the noiseless case.

We conclude with some future directions. As discussed in Section~\ref{sec:sbm} and further suggested by our simulations, there remains a gap between our upper and lower bounds in certain regimes of the general SBIM, due to the difficulty of bounding $H\lp X_{\calC_1},...,X_{\calC_m}\rp$ when $\lbp X_{\calC_1},...,X_{\calC_m} \rbp$ are not mutually independent. Other directions of interest include designing non-adaptive group testing schemes for our setting, %
deriving bounds under other noise models such as dilution and Z-channel models, and extending our infection model to longer time horizons (e.g., SIR or SIS-type infection models from the epidemiology literature). Finally, characterizing the complexity of group testing under the general graph-based infection model described in Section~\ref{subsec:SBM_equivalence} (beyond the SBIM studied in this paper) is fertile ground for future work.

    \section*{Acknowledgments} 
    This work was supported in part by NSF Grant \#1817205, the Center for Science of Information (CSoI), an NSF Science and Technology Center under grant agreement CCF-0939370, a Cisco Systems Stanford Graduate Fellowship, and a National Semiconductor Corporation Stanford Graduate Fellowship. We thank the anonymous reviewers for their invaluable suggestions which helped us improve this paper. 
	 
	\bibliographystyle{ieeetr}
	\bibliography{references}
	
	\newpage 
	\appendix

\subsection{Proof of Lemma~\ref{lemma:sbm_LB}}\label{app:cond_entropy_LB}
Let $S_i$ be the indicator variable of whether the $i^\text{th}$ individual is a seed. By standard information-theoretic arguments, we have 
\begin{align*}
    H(X_1,\ldots,X_n) &= I(X_1,\ldots, X_n \, ; \, S_1,\ldots, S_n) + H(X_1,\ldots, X_n \, | \, S_1,\ldots, S_n) \\
    &= H(S_1,\ldots, S_n) - H(S_1,\ldots,S_n\,|\,X_1,\ldots,X_n) + H(X_1,\ldots,X_n\,|\,S_1,\ldots,S_n) \\
    &= n\cdot H(S_1) - \sum_{i=1}^n H(S_i \, | \, X_1,\ldots,X_n, S_1,\ldots, S_{i-1}) + H(X_1,\ldots,X_n\,|\,S_1,\ldots,S_n) \\
    &\geq n\cdot H(S_1) - n\cdot H (S_1\,|\,X_1) + H(X_1,\ldots,X_n\,|\,S_1,\ldots,S_n) \\
    &= n\cdot I(X_1;S_1) + H(X_1,\ldots,X_n\,|\,S_1,\ldots,S_n).
\end{align*}
This proves \eqref{eqn:full_LB} in Lemma~\ref{lemma:sbm_LB}.

Next, we prove \eqref{eqn:cond_entropy_LB}. First, note that $H(X) \geq  H\left( X_1,\ldots,X_n \, \middle \vert S_1,\ldots,S_n \, \right)$ since $I(X_1;S_1) \geq 0$. (This also follows directly from the fact that conditioning reduces entropy.) Next, using the shorthand notation $X = (X_1,\ldots,X_n)$, $S = (S_1,\ldots, S_n)$, we have
$$ %
H\left( X \, \middle \vert \, S \right) = \sum_{s\in \left\{0,1\right\}^n} \PP\left( S = s\right)\cdot H\left( X \, \middle \vert \, S=s \right). $$
Observe that after conditioning on the locations of the seeds, $X_1,...,X_n$ are mutually independent. Moreover, for $i \in \calC_\ell$, the marginal distribution of $X_i$ can be specified as follows:
$$ \PP\left( X_i = 1 \, \middle \vert \, S = s \right) = 
\begin{cases} 
1, & \text{ if } s_i = 1,\\
1-(1-q_1)^{\sum_{j \in \calC_\ell} s_j}\lp1-q_2\rp^{\sum_{j\not\in \calC_\ell}s_j}, & \text{ if } s_i = 0.
\end{cases} $$
Writing $z_\ell \triangleq \sum_{j \in \calC_\ell} s_j$, 
the conditional entropy is given by
$$ H\left( X\, \middle \vert \, S=s \right) = \sum_{\ell = 1}^m \left( k - z_\ell\right)\cdot \mathsf{h}_\mathsf{b}\left( 1-(1-q_1)^{z_\ell}\lp 1- q_2 \rp^{\sum_{\ell'\neq \ell}z_{\ell'}} \right), $$
where $\mathsf{h}_\mathsf{b}\left( \cdot \right)$ is the binary entropy function.
Since $S_i \diid \msf{Ber}(p)$, we have $Z_\ell \diid \msf{Binom}(k, p)$ and hence
\begin{equation} \label{eq:bino_expression}
H\left( X \,\middle \vert \, S \right) = \EE_{Z, Z'} \left[m \cdot \left( k - Z\right)\cdot \mathsf{h}_\mathsf{b}\left( 1-(1-q_1)^Z(1-q_2)^{Z'}\right)\right], 
\end{equation}
where $Z \sim \mathsf{Binom}\left( k, p\right)$ and $Z' \sim \mathsf{Binom}\left( n-k, p\right)$.

\QED

\subsection{Proof of Theorem~\ref{thm:sbm_LB}} \label{app:sbm_LB}
In light of Lemma~\ref{lemma:sbm_LB}, it suffices to show that the following lower bound holds under the assumptions in the theorem statement: 
\begin{align*}
\EE_{Z, Z'} \Big[( k - Z)\cdot & \mathsf{h}_\mathsf{b}\Big( 1-(1-q_1)^Z\lp 1- q_2 \rp^{Z'}\Big) \Big] \succeq  mk^2 pq_2\log\lp \frac{1}{npq_2} \rp + k^2p q_1 \log\lp \frac{1}{q_1 + npq_2} \rp,
\end{align*} 
where $Z \sim \mathsf{Binom}\left( k, p\right)$ and $Z'\sim \msf{Binom}(n-k, p)$.
Our proof leverages the concentration of $Z$ and $Z'$ around their means. 

First we assume $n\cdot p \cdot q_2 \preceq 1$, and let $\epsilon \in (0,1)$ be a value to be specified. Define $$ z^* \eqDef \frac{1/2 - np(1+\epsilon)q_2}{q_1}. $$
Then as long as $Z$ and $Z'$ satisfy the following two conditions
\begin{enumerate}
    \item $ \lbp np(1-\epsilon) \leq Z' \leq np(1+\epsilon) \rbp$, 
    \item $ Z \leq z^*$, 
\end{enumerate}
we have 
\begin{equation}
    \frac{1}{2} \geq Z\cdot q_1+ Z'\cdot q_2 \geq 1-\lp1-q_1\rp^Z\lp 1-q_2 \rp^{Z'}.
\end{equation}
Since $1-\lp1-q_1\rp^Z\lp 1-q_2 \rp^{Z'}$ is an increasing function of $Z$ and $Z'$,
$ \msf{h_b}\lp 1-\lp1-q_1\rp^Z\lp 1-q_2 \rp^{Z'} \rp $ must increase with $Z$ and $Z'$ if they satisfy the above conditions. Therefore, we have
\begin{align}\label{eq:thm4_2}
    & \EE_{Z, Z'}\lb (k-Z)\msf{h_b}\lp 1-(1-q_1)^Z\lp 1- q_2 \rp^{Z'} \rp \rb\nonumber\\
    \geq & \EE_{Z, Z'}\lb (k-Z)\msf{h_b}\lp 1-(1-q_1)^Z\lp 1- q_2 \rp^{Z'}\rp\cdot\bbm{1}_{\lbp 0\leq Z\leq z^* \rbp}\cdot \bbm{1}_{\lbp np(1-\epsilon) \leq Z' \leq np(1+\epsilon)\rbp} \rb \nonumber\\
    \geq & \underbrace{\EE_{Z, Z'}\lb (k-Z)\msf{h_b}\lp 1-\lp 1- q_2 \rp^{Z'}\rp\cdot\bbm{1}_{\lbp Z = 0 \rbp}\cdot \bbm{1}_{\lbp np(1-\epsilon) \leq Z' \leq np(1+\epsilon)\rbp} \rb}_{\text{(a)}} + \nonumber\\
    &\underbrace{\EE_{Z, Z'}\lb (k-Z)\msf{h_b}\lp 1-(1-q_1)^Z\lp 1- q_2 \rp^{Z'}\rp\cdot\bbm{1}_{\lbp 1\leq Z\leq z^* \rbp}\cdot \bbm{1}_{\lbp np(1-\epsilon) \leq Z' \leq np(1+\epsilon)\rbp} \rb}_{\text{(b)}}.
\end{align}
We will pick $\epsilon = \frac{1}{2}$. Then (a) can be bounded by
\begin{align*}
    \text{(a)} &\geq k\cdot \msf{h_b}\lp q_2\cdot np(1-\epsilon) - \lp q_2\cdot np(1-\epsilon) \rp^2 \rp \lp 1-2\cdot\exp\lp -\frac{n\epsilon^2p}{3} \rp \rp \\
    & \succeq k\lp npq_2(1-\epsilon)\log\lp \frac{1}{npq_2(1-\epsilon)} \rp \lp 1-2\cdot\exp\lp -\frac{n\epsilon^2p}{3} \rp\rp\rp  \\
    & \succeq k\lp npq_2\log\lp \frac{1}{npq_2} \rp \rp
\end{align*}
 where in the first inequality we use
 \begin{enumerate}
     \item $Z' \geq np(1-\epsilon)$
     \item $(1-q_2)^{Z'} \leq e^{-q_2\cdot Z'} \leq 1-q_2\cdot Z' +\lp q_2\cdot Z' \rp^2$
     \item Chernoff bound on $Z'$, 
 \end{enumerate}
and in the third inequality we assume $np \succeq 1$.
Next, (b) can be bounded by
\begin{align*}
    \text{(b)} &\geq \msf{h_b}\lp q_1 + npq_2(1-\epsilon) - \lp q_1 + npq_2(1-\epsilon) \rp^2 \rp \cdot \EE_Z\lb (k-Z)\bbm{1}_{\lbp 1\leq Z \leq z^* \rbp} \rb \cdot \lp 1-2\cdot\exp\lp -\frac{n\epsilon^2p}{3} \rp \rp \\
    & \succeq \lp q_1 + npq_2 \rp \log\lp \frac{1}{q_1 + npq_2} \rp\cdot \EE_Z\lb (k-Z)\bbm{1}_{\lbp 1\leq Z \leq z^* \rbp} \rb.
\end{align*}
We will now lower bound $\EE_Z\lb (k-Z)\bbm{1}_{\lbp 1\leq Z \leq z^* \rbp} \rb$ as in Theorem~\ref{thm:k-clique_LB}. Observe that
\begin{align}\label{eq:thm4_1}
    \EE_Z\lb (k-Z)\bbm{1}_{\lbp 1\leq Z \leq z^* \rbp} \rb 
    & \geq \EE_Z\lb k-Z \rb - k \cdot \PP\lbp Z = 0 \rbp - k \cdot \PP\lbp Z\geq z^* \rbp \nonumber\\
    & \geq k\lp 1- p -(1-p)^k - \PP\lbp Z\geq z^* \rbp \rp \nonumber\\
    & \succeq k\lp kp- \PP\lbp Z\geq z^* \rbp \rp.
\end{align}
Finally, applying Hoeffding's inequality to $\PP\lbp Z\geq z^* \rbp$ yields
\begin{align*}
    \PP\lbp Z\geq z^* \rbp 
    & \leq \exp\lp -2k\lp p - \frac{z^*}{k} \rp^2 \rp = \exp\lp -2k\lp p - \frac{\frac{1}{2} - npq_2(1+\epsilon)}{q_1k} \rp^2 \rp \\
    & \overset{\text{(1)}}{\preceq}  \exp\lp -2k\lp \frac{1}{2q_1k} \rp^2 \rp = \exp\lp -\frac{1}{2kq_1^2} \rp \overset{\text{(2)}}{\leq} \frac{kp}{2},
\end{align*}
where in (1) we use the facts that 1) $n\cdot p\cdot q_2 \preceq 1$ and 2) $p\preceq \frac{1}{q_1k}$, and (2) holds when 
$$ q_1 \leq \frac{1}{\sqrt{2k\cdot\lp \log\lp \frac{1}{kp} \rp+1\rp}}. $$
Plugging into \eqref{eq:thm4_1} yields
\begin{equation}
    \EE_Z\lb (k-Z)\bbm{1}_{\lbp 1\leq Z \leq z^* \rbp} \rb \succeq k^2p,
\end{equation}
and thus by putting together our bounds on (a) and (b) in \eqref{eq:thm4_2}, we arrive at
\begin{align}
    &\EE_{Z, Z'}\lb (k-Z)\msf{h_b}\lp 1-(1-q_1)^Z\lp 1- q_2 \rp^{Z'} \rp \rb \\
    \geq &k\lp npq_2\log\lp \frac{1}{npq_2} \rp \rp + k^2p\cdot\lp q_1 + npq_2 \rp \log\lp \frac{1}{q_1 + npq_2} \rp \\
    \geq & mk^2 pq_2\log\lp \frac{1}{npq_2} \rp + k^2p\cdot q_1 \log\lp \frac{1}{q_1 + npq_2} \rp.
\end{align}

\QED

\subsection{Proof of Lemma~\ref{lem:sbm_marginal}}\label{app:sbm_marginal} 

Let $S_v$ be the indicator random variable of whether an individual $v$ is a seed, and assume without loss of generality that $v \in \calC_1$. 
We have 
		\begin{align}
		\PP(X_v = 1) &= \underbrace{\PP(X_v = 1\, | \, S_v = 1)}_{=1}\cdot \underbrace{\PP(S_v = 1)}_{=p} +  \PP\lp X_v = 1 \, \middle\vert \, S_v = 0 \rp)\cdot \PP(S_v = 0) \nonumber \\
		&= p + (1-p) \cdot \PP\lp X_v = 1 \, \middle\vert \, S_v = 0 \rp. \nonumber %
		\end{align}
		Given that $v$ is not a seed, $X_v = 1$ if and only if $v$ is infected by another seed. Hence, 
		\begin{align*}
		\PP \lp X_v = 1 \, \middle\vert \, S_v = 0 \rp &= \PP\bigl( \lbp v \text{ is infected by another individual}\rbp \bigr) \\
		&= 1- \prod_{u \in [n]} \PP\bigl( \lbp v \text{ isn't infected by } u \rbp \bigr) \\
		&= 1- \prod_{u \in [n]} \Bigl( 1-\PP\bigl( \lbp v \text{ is infected by } u \rbp \bigr) \Bigr)  \\
		&= 1 - \prod_{u \in [n]} \Bigl( 1-\PP \lp \lbp v \text{ is infected by } u \rbp \, \middle\vert \, S_u = 1 \rp \cdot \PP(S_u = 1 ) \Bigr)   \\
		&= 1 - \Bigg(\prod_{u \in \calC_1 \backslash \{v\}} (1-p \cdot q_1)\Bigg)\cdot \Bigg(\prod_{w \not\in \calC_1} (1-p \cdot q_2)\Bigg) \\
		&=  1- (1-p\cdot q_1)^{k-1}\cdot (1-p\cdot q_2)^{n-k}.
		\end{align*}

\QED 

\subsection{Proof of Lemma~\ref{lem:sbm_clique_marginal}}\label{app:sbm_clique_marginal}
Let $\calA$ be the event that no member of community $\calC_1$ is selected as a seed, and let $\calB$ be the event that some member of $\calC_1$ is infected by an individual outside $\calC_1$. We further denote by $\calB_u$ the event that vertex $u$ infects some member of $\calC_1$, where $u \not\in \calC_1$. Note that $X_{\calC_1} = 1$ if and only if either $\calA^c$ occurs or $\calA \cap \calB$ occurs. Moreover, $\calA$ and $\calB$ are independent events. We have that $\PP(\calA) = (1-p)^k$, and thus 
\begin{align*}
    \PP(X_{\calC_1} = 1) = \PP(\calA^c) + \PP(\calA)\cdot \PP(\calB)
    &= 1- (1-p)^k + (1-p)^k \cdot \PP(\calB) \\
    &= 1 - (1-p)^k \cdot (1-\PP(\calB)).
\end{align*}

Finally, we compute $\PP(\calB)$ as 
\begin{align*}
    \PP(\calB) &= 1 - \prod_{u \not\in \calC_1} \PP(\calB_u^c) \\
    &= 1 - \prod_{u \not\in \calC_1} \Big(\PP(\calB_u^c\, | \, S_u = 1)\cdot \underbrace{\PP(S_u = 1)}_{=p} + \underbrace{\PP(\calB_u^c \, | \, S_u = 0)}_{= 1} \cdot \underbrace{\PP(S_u = 0)}_{=1-p} \Big) \\
    &= 1 - \prod_{u \not\in \calC_1} \Big(1-p + p \cdot \PP(\calB_u^c \, | \, S_u = 1)\Big) \\
    &= 1 - \prod_{u \not\in \calC_1} \Big(1-p + p \cdot (1-q_2)^k \Big) \\
    &= 1 - \Bigg(1 - p\cdot \Big(1-(1-q_2)^k\Big)\Bigg)^{n-k}.
\end{align*}

\QED

\subsection{Proof of Lemma~\ref{lem:k-clique_LB}} \label{app:k-clique_LB}
Let $f\left( q\right) = \frac{\log\left( q \right)}{\log\left(1-q\right)}$, so that $f(q)$ solves $1-(1-q)^Z = 1-q$. Then %
\begin{align}\label{eq:lb}
    \EE_Z \left[\left( k - Z\right)\cdot \mathsf{h}_\mathsf{b}\left( 1-(1-q)^Z\right)\right]
    &\geq \EE_Z \left[\left( k - Z\right)\cdot \mathsf{h}_\mathsf{b}\left( 1-(1-q)^Z\right)\cdot \mathbbm{1}_{\left\{1\leq Z \leq f(q) \right\} }\right]\nonumber \\
    &\overset{\text{(a)}}{\geq} \mathsf{h}_\mathsf{b}\left( q\right)\cdot\EE_Z \left[\left( k - Z\right)\cdot \mathbbm{1}_{\left\{1\leq Z \leq f(q) \right\} }\right]\nonumber\\
    &\geq \mathsf{h}_\mathsf{b}(q) \left(\EE_Z \left[k - Z\right]- k\cdot \PP\left\{Z=0\right\}- k\cdot \PP\left\{Z > f(q) \right\} \right) \nonumber\\
    &= k\cdot\mathsf{h}_\mathsf{b}(q)\left( (1-p)\left(1-(1-p)^{k-1}\right)-\PP\left\{Z>f(q) \right\}\right)\nonumber\\
    &\overset{\text{(b)}}{\geq} k\cdot\mathsf{h}_\mathsf{b}(q)\left( (1-p)\left( (k-1)p - (k-1)^2p^2\right)-\PP\left\{Z>f(q) \right\}\right) \nonumber\\
    & \overset{\text{(c)}}{\succeq} \frac{k}{2}\cdot\mathsf{h}_\mathsf{b}(q)\left( k\cdot p-\PP\left\{Z>f(q) \right\}\right),
\end{align}
where (a) is due to the fact that $\msf{h_b}(x) \geq \msf{h_b}(q)$ for all $q \leq x \leq 1-q$, (b) holds since $(1-p)^r \leq e^{-pr}$ and $e^x \leq 1+x+x^2$ for $x \leq 1$, and (c) is due to the assumption $p\preceq 1/k$.

We then upper bound $\PP\left\{Z>f(q) \right\}$ by Hoeffding's inequality:
\begin{align*}
    \PP\left\{Z>f(q)\right\}
     \leq \exp\left( -2k\left( p- \frac{f(q)}{k}\right)^2 \right) 
     \overset{\text{(a)}}{\preceq} \exp\left( -2k\left( \frac{f(q)}{2k}\right)^2 \right)
     \leq \exp\left( -\frac{f(q)^2}{2k}\right) 
    \overset{\text{(b)}}{\preceq} \frac{kp}{2},
\end{align*}
where (a) holds by the assumption $k\cdot p\cdot q \preceq 1$, so that
$$ k\cdot p\preceq \frac{q}{2}\log\left(\frac{1}{q}\right) \leq \frac{q}{1-q}\log\left( \frac{1}{q}\right) \leq f(q),$$
and (b) holds due to the assumption $q \preceq \frac{1}{\sqrt{k}\cdot\sqrt{\log\left( \frac{1}{k\cdot p} \right)}}$. Plugging into \eqref{eq:lb} yields
$$ \EE_Z \left[\left( k - Z\right)\cdot \mathsf{h}_\mathsf{b}\left( 1-(1-q)^Z\right)\right] \succeq k^2\cdot p \cdot q\cdot\log\lp \frac{1}{q} \rp \succeq k^2\cdot p\cdot q \cdot \lp\log k + \log\log\lp \frac{1}{kp} \rp\rp, $$
where in the last inequality we use the assumption $q \preceq \frac{1}{\sqrt{k}\cdot\sqrt{\log\left( \frac{1}{k\cdot p} \right)}}$ again.

\QED

\subsection{Proof of Theorem~\ref{thm:sbm_graphaware}}\label{app:sbm_graphaware}
    Let $T_1$ and $T_2$ be the number of tests performed, respectively, in Step 2 and Step 3 of the community-aware algorithm. Specifically, $T_1$ is equal to the number of tests used by binary splitting to identify the infected communities, and $T_2$ is  the number of tests to identify infected individuals within each infected community. Note that $T = T_1 + T_2$. We will bound $\EE[T_1]$ and $\EE[T_2]$ separately.
    
	Let $N$ be the number of infected communities. By Lemma~\ref{lem:sbm_clique_marginal}, we have 
	\[\EE[N] = \frac{n}{k}\cdot \PP(X_{\calC_1} = 1) = \frac{n}{k}\cdot \Bigg(1 - (1-p)^k \cdot \Bigg(1 - p\cdot \Big(1-(1-q_2)^k\Big)\Bigg)^{n-k}\Bigg).\]
	Taking Lemma~\ref{lem:binary_split} with $n = n/k$ and $\alpha = N$ gives 
	\[T_1 \leq (\log_2(n/k) + 2) \cdot N + 1\]
	so that 
	\[\EE[T_1] \leq \frac{n}{k}\cdot \Big(\log_2(n/k) + 1\Big)\cdot \Bigg(1 - (1-p)^k \cdot \Bigg(1 - p\cdot \Big(1-(1-q_2)^k\Big)\Bigg)^{n-k}\Bigg) + 1.\]
	For the second stage of the algorithm, let $Z_i$ denote the number of tests used by binary splitting to identify all infected members of the $i^\text{th}$ community. Since $T_2 = \Sum{i=1}{n/k} Z_i \cdot \mathbbm{1}_{\lbp X_{\calC_i} = 1 \rbp}$, we have 
	\begin{align*}
		\EE[T_2] &= \Sum{i=1}{n/k} \EE \lb Z_i \cdot \mathbbm{1}_{\lbp X_{\calC_i} = 1\rbp } \rb = \frac{n}{k} \cdot  \EE \lb Z_1 \cdot \mathbbm{1}_{\lbp X_{\calC_1} = 1\rbp } \rb = \frac{n}{k}\cdot \PP(X_{\calC_1} = 1) \cdot  \EE\lb Z_1 \, \middle\vert \, X_{\calC_1} = 1 \rb  \\
	\end{align*}
	
	Let $M$ denote the number of infected members of $\calC_1$. Then by Lemma~\ref{lem:binary_split},  
    $$\EE\lb Z_1 \, \middle\vert \, X_{\calC_1} = 1 \rb \leq (\log_2k + 2) \cdot \EE \lb M \, \middle\vert \, X_{\calC_1} = 1 \rb + 1$$
    and, assuming without loss of generality that $\calC_1 = [k]$, 
    \begin{align*}
        \EE\lb M \, \middle\vert \, X_{\calC_1} = 1 \rb &= \Sum{j=1}{k} \PP\lp X_j = 1 \, \middle\vert \, X_{\calC_1} = 1 \rp \\
        &= k\cdot \PP \lp X_1 = 1 \, \middle\vert \, X_{\calC_1} = 1 \rp \\
        &= k \cdot \frac{\PP(X_1 = 1, \, X_{\calC_1} = 1)}{\PP(X_{\calC_1} = 1) } \\
        &= k \cdot \frac{\PP(X_1 = 1)}{\PP(X_{\calC_1} = 1)} \\
        &= k \cdot \frac{1 - (1-p) \cdot (1-p \cdot q_1)^{k-1} \cdot (1-p \cdot q_2)^{n-k}}{\PP(X_{\calC_1} = 1)}
    \end{align*}
    where in the last line we invoke Lemma~\ref{lem:sbm_marginal}.
	Putting everything together gives 
	\[\EE[T_2] \leq n \cdot (\log_2 k + 2) \cdot \Big(1 - (1-p) \cdot (1-p \cdot q_1)^{k-1} \cdot (1-p \cdot q_2)^{n-k}\Big) + \frac{n}{k}\cdot \Bigg(1 - (1-p)^k \cdot \Bigg(1 - p\cdot \Big(1-(1-q_2)^k\Big)\Bigg)^{n-k}\Bigg)\]
	and therefore 
	\begin{align*}
    \EE[T] \leq \frac{n}{k} \cdot &\Big(\log_2(n/k) + 3\Big) \cdot \Bigg(1 - (1-p)^k \cdot \Bigg(1 - p\cdot \Big(1-(1-q_2)^k\Big)\Bigg)^{n-k}\Bigg) \\
&+ n \cdot \Big(\log_2 k + 1\Big) \cdot \Big(1 - (1-p) \cdot (1-p \cdot q_1)^{k-1} \cdot (1-p \cdot q_2)^{n-k}\Big) + 1\\
\end{align*}  
	
\QED

\end{document}